\documentclass[11]{article}
\usepackage[margin=1in]{geometry}
\usepackage[utf8]{inputenc}
\usepackage{DP-macros}
\usepackage{amsmath,amssymb,amsthm}
\usepackage{algorithm,algpseudocode}
\usepackage{hyperref}

\newtheorem{theorem}{Theorem}
\newtheorem{lemma}[theorem]{Lemma}

\newtheorem{definition}[theorem]{Definition}

\newcommand{\jl}{T}

\newcommand{\dist}{\mu}
\DeclareMathOperator{\vol}{vol}
\DeclareMathOperator{\rad}{rad}
\DeclareMathOperator{\diam}{diam}
\DeclareMathOperator{\volLB}{volLB}
\DeclareMathOperator{\vrad}{vrad}
\DeclareMathOperator{\tr}{tr}
\newcommand{\sym}{\text{sym}}
\newcommand{\sep}{P}
\newcommand{\jldim}{\ell}
\newcommand{\ip}[2]{\langle #1, #2\rangle }
\newcommand{\cut}[1]{}

\title{Private Query Release via the Johnson-Lindenstrauss Transform}
\author{Aleksandar Nikolov \\
Univesity of Toronto\\
anikolov@cs.toronto.edu}
\date{}
\begin{document}

\maketitle

\begin{abstract}
  We introduce a new method for releasing answers to statistical queries with differential privacy, based on the Johnson-Lindenstrauss lemma. The key idea is to randomly project the query answers to a lower dimensional space so that the distance between any two vectors of feasible query answers is preserved up to an additive error. Then we answer the projected queries using a simple noise-adding mechanism, and lift the answers up to the original dimension. Using this method, we give, for the first time, purely differentially private mechanisms with optimal worst case sample complexity under average error for answering a workload of $\qsize$ queries over a universe of size $\usize$. As other applications, we give the first purely private efficient mechanisms with optimal sample complexity for computing the covariance of a bounded high-dimensional distribution, and for answering 2-way marginal queries. We also show that, up to the dependence on the error, a variant of our mechanism is nearly optimal for every given query workload.
\end{abstract}

\section{Introduction}

One of the central problems in private data analysis is to release
answers to statistical (also known as linear) queries on the
data. Here, a statistical query is defined by a function
\(\query:\uni \to \R\), and the value of the query on a dataset
\(\ds \in \uni^\dsize\) is simply its average over data points, which
we write as
\(\query(\ds) = \frac{1}{\dsize}\sum_{\dsrow \in
  \ds}{\query(\dsrow)}\), with the sum taken with multiplicity.
Counting queries, which ask what fraction of the dataset satisfies
some predicate, are a special case of statistical queries, and are
themselves of significant interest. Many organizations report summary
statistics in the form of tables, where each cell in the table is the
answer to a counting query. This is the case, for example, for many of
the tables released by official statistics agencies such as the US
Census Bureau: see~\cite{HDMM} for a worked out example. As other examples,
statistical queries capture CDFs of one-dimensional distribution and
higher dimensional generalizations, mean estimation, and loss
gradients in empirical risk minimization. Thus, when we analyze data
about people whose privacy needs to be protected, the question of
privately computing answers to statistical queries is of fundamental
interest.

The privacy framework we adopt is
differential privacy~\cite{DworkMNS06}, which provides strong semantic
guarantees for protecting the privacy of individuals represented in a
dataset. We say that a randomized algorithm \(\mech\)  (usually called a
mechanism) that takes datasets \(\ds \in \uni^\dsize\) is \((\eps,
\delta)\)-differentially private if for any two datasets \(\ds,\ds'\)
that differ in at most one element, and for any measurable event \(S\)
in the range of \(\mech\), we have
\[
  \Pr(\mech(\ds) \in S) \le e^\eps \Pr(\mech(\ds')\in S) + \delta.
\]
The setting when \(\delta = 0\) is usually called pure differential
privacy, and the setting when \(\delta > 0\) is called approximate
differential privacy. In this paper, we focus on purely differentially
private mechanisms, although our basic technique applies to 
approximate differential privacy, as well as to intermediate privacy
notions~\cite{DworkR16,BunS16,Mironov17,GDP}. Pure differential privacy is the strictest of the standard
variants of differential privacy, and has some nice properties not
shared by approximate differential privacy and other relaxations: it
is a single parameter definition; satisfies a simple and tight
composition theorem (composing private mechanisms just requires ``adding
the epsilons''); it satisfies a strong group privacy property, i.e.,
some measure of privacy protection is automatically offered to small
groups of people and not just to individuals. These benefits, on the
other hand, can
sometimes come at the cost of increased error. Understanding when this
is the case requires completely understanding how much error a purely
differentially private algorithm needs to introduce for a given
statistical task.

Coming back to statistical queries, we can see from the examples in
the first paragraph that, typically, we are interested in releasing
answers to many statistical queries, rather than in just a single query. Let us call a sequence \(\queries = (\query_1, \ldots, \query_\qsize)\) of statistical queries a workload. The following are some of the fundamental questions one can ask about releasing answers to a workload of statistical queries under differential privacy:
\begin{enumerate}
    \item How much error is necessary, in the worst case, to answer any workload of \(\qsize\) statistical queries over a universe \(\uni\) of size \(\usize\) given a dataset of size \(\dsize\)?
    
    \item How can we characterize the optimal trade-off between error and dataset size for a given  workload \(\queries\)?
    
    \item Can we achieve the optimal worst case error in question 1., or the optimal error vs dataset size tradeoff in question 2.~using an efficient mechanism?
\end{enumerate}
There are many variants of these questions, depending on what flavor
of differential privacy we adopt (pure, or approximate, or
concentrated differential privacy), and what measure of error we
choose (worst case error, or average, i.e., root mean squared
error). Yet, there are few settings in which precise answers are
known. A tight answer for question 1.~in the setting of approximate
differential privacy was given by Bun, Ullman, and
Vadhan~\cite{BunUV14}. Question 2.~is open for any variant of
differential privacy and measure of error that has been studied, but
there have been some partial approximate
characterizations~\cite{HardtT10,BhaskaraDKT12,NTZ,Nikolov15,cdp}. The
answer to question 3.~depends on how the workload is represented, but
is generally wide open: for example, we do not even know an efficient
mechanism that achieves the optimal error vs dataset size tradeoff for
width-\(3\) marginal queries. 

In this paper we focus on the most basic of the three questions,
question 1., in the case of pure differential privacy. Surprisingly,
in this setting, the question is open, both for worst case error and
for average error. We introduce a new algorithmic technique for query
release, based on the Johnson-Lindenstrauss transform, and we show
that it resolves question 1.~for pure differential privacy and average
error.

In order to discuss our results and prior work more precisely, we now introduce some notation. Let us use \(\queries(\ds)\) to denote the vector of query answers \((\query_1(\ds), \ldots, \query_\qsize(\ds))\). Let us say the queries in the workload are bounded if, for each \(i\in [\qsize]\) and each \(\dsrow\in \uni\), \(|\query_i(\dsrow)| \le 1\). We define, respectively, the worst case (or \(\ell_\infty\)), and the average (or root mean squared, or \(\ell_2\)) error of a mechanism \(\mech\) on a workload \(\queries\) and datasets of size \(\dsize\) as 
\begin{align*}
&\err^{\ell_\infty}(\mech,\queries,n) = \max_{\ds \in \uni^\dsize} \E \|\mech(\ds) - \queries(\ds)\|_\infty;
&\err^{\ell_2}(\mech,\queries,n) = \max_{\ds \in \uni^\dsize} \E\left[ \frac{1}{\sqrt{\qsize}}\|\mech(\ds) - \queries(\ds)\|_2\right].
\end{align*}
Above, we assume that on input \(\ds\), the mechanism \(\mech\) outputs a random \(\qsize\)-dimensional vector \(\mech(\ds)\), and the expectation is taken with respect to the mechanism's randomness. The sample complexity of the mechanism on the workload \(\queries\) is the smallest dataset size for which the mechanism can achieve error at most \(\alpha\), denoted
\(
    \sc^{\ell_p}(\mech, \queries,\alpha) = \inf\{n: \err^{\ell_p}(\mech,\queries,n) \le \alpha\},
\)
where \(p \in \{2, \infty\}\).
Note that \(\err^{\ell_2}(\mech,\queries,n) \le
\err^{\ell_\infty}(\mech,\queries,n)\), and, therefore,
\(\sc^{\ell_2}(\mech, \queries,\alpha) \le \sc^{\ell_\infty}(\mech,
\queries,\alpha)\). Finally, we define the optimal sample complexity
of a workload \(\queries\) under \(\eps\)-differential privacy
by \[\sc^{\ell_p}_\eps(\queries,\alpha) = \inf\{\sc^{\ell_p}(\mech,
  \queries,\alpha): \mech \text{ is } \eps \text{-differentially
    private}\}.\] Now, question 1.~can be formulated precisely as the
problem of giving tight bounds on the worst case value of
\(\sc^{\ell_p}_\eps(\queries,\alpha)\) over workloads \(\queries\) of
\(\qsize\) bounded queries on a universe of size \(\usize\). The
Laplace noise mechanism~\cite{DworkMNS06}, a variant of the K-norm
mechanism~\cite{SteinkeU17}, and the SmallDB mechanism~\cite{BLR} show that\footnote{Here and in the rest of the paper we use the notation \(A \lesssim B\) to signify that there exists an absolute constant \(C > 0\) such that \(A \le CB\). We use \(A \gtrsim B\) to denote \(B \lesssim A\).}
\begin{align*}
&\sc^{\ell_2}_\eps(\queries, \alpha) \lesssim 
     \min\left\{\frac{\qsize}{\eps \alpha}, \frac{\log(\usize)}{\eps\alpha^3}\right\};
    &\sc^{\ell_\infty}_\eps(\queries, \alpha) \lesssim 
     \min\left\{\frac{\qsize}{\eps \alpha}, \frac{\log(\qsize)\log(\usize)}{\eps\alpha^3}\right\}.
\end{align*}
Stronger bounds are known using methods tailored to average error. The K-norm mechanism,\footnote{The first bound is not stated in~\cite{HardtT10}, but follows directly from their results and known estimates of the volume of a polytope contained in a ball.} and a combination of the K-norm mechanism and the projection mechanism~\cite{NTZ} give the  bounds
\[
\sc^{\ell_2}_\eps(\queries, \alpha) \lesssim 
\min \left\{
\frac{\log^{3/2}(\qsize)\sqrt{\qsize\log(\usize)}}{\eps\alpha},
\frac{\log^2(\qsize)\log^{3/2}(\usize)}{\eps\alpha^2}
\right\}.
\]
The first term on the right can be boosted to a sample
complexity upper bound for worst case error using non-private boosting
for queries (see Section~6.1 in~\cite{NTZ}).
We do not know how to boost the second term, since the error per query
is private information, and private boosting incurs too big of a loss
in the sample complexity in the pure differential privacy setting. The
second term is the only known sample complexity bound in the pure differential privacy setting which has quadratic dependence on \(\frac{1}{\alpha}\), and polylogarithmic dependence on the number of queries and universe size. Unfortunately, the powers of the logarithmic terms are not known to be tight. In particular, the best known lower bounds on sample complexity come from packing arguments~\cite{HardtT10,Hardt-thesis,De12}, and show that there exists a workload \(\queries\) of \(\qsize\) bounded queries on a universe of size \(\usize\) such that
\begin{align*}
&\sc^{\ell_2}_\eps(\queries, \alpha) \gtrsim 
     \min\left\{\frac{\qsize}{\eps \alpha},
    \frac{\sqrt{\qsize\log \usize}}{\eps \alpha},
    \frac{\log(\usize)}{\eps\alpha^2}\right\};
    &\sc^{\ell_\infty}_\eps(\queries, \alpha) \gtrsim 
     \min\left\{\frac{\qsize}{\eps \alpha},
    \frac{\sqrt{\qsize\log \usize}}{\eps \alpha},\frac{\log(\qsize)\log(\usize)}{\eps\alpha^2}\right\}.
\end{align*}
We note that these appear to be the best lower bounds that can be proved with packing arguments. Thus, closing the gap between upper and lower bounds would require either new algorithmic ideas, or developing a lower bound method for pure differential privacy that goes beyond packing. 

\paragraph{Results.} In this paper, we introduce a simple new algorithmic idea for the query release problem. While our main technique applies to any flavor of differential privacy, we focus on its application in the setting of pure differential privacy, where it gives a complete answer to question 1.~above for average error. In particular, we prove the following theorem.

\begin{theorem}\label{thm:main-worst-case}
For any workload \(\queries\) of \(\qsize\) bounded queries on a universe of size \(\usize\), we have
\begin{equation}
    \label{eq:sc-opt}
    \sc^{\ell_2}_\eps(\queries,\alpha)
    \lesssim
    \min\left\{
    \frac{\qsize}{\eps \alpha},
    \frac{\sqrt{\qsize \log \usize}}{\eps \alpha},
    \frac{\log \usize}{\eps \alpha^2}
    \right\}.
\end{equation}
Moreover, this sample complexity is achieved by a mechanism running in time polynomial in \(\qsize,\usize,\dsize,\frac{1}{\eps}\).
\end{theorem}
The sample complexity bounds in \eqref{eq:sc-opt} are tight for random
queries (see Section~\ref{sect:rand-lb}), and resolve Open Problem 2.~from~\cite{DPorg-open-problem-optimal-query-release}.  The first term on the right
hand side of \eqref{eq:sc-opt} is known, and can be shown using either
the Laplace noise or the K-norm mechanism. To get the second term, we
optimize the privacy budget of the K-norm mechanism in order to remove
unnecessary logarithmic terms. Our main contribution is in
establishing the last term by giving a new private query release
mechanism, which we call the JL mechanism. The key new idea is to use a variant of
the Johnson-Lindenstrauss lemma to first randomly project the query
answers to a lower dimensional space, then add noise in the lower
dimensional space, and finally lift up the answers to the original
dimension.
To explain why this gives the sample complexity in \eqref{eq:sc-opt}, let us define the sensitivity polytope
\(\spoly_\queries\) of the workload \(\queries\), given by
\(\spoly_\queries = \mathrm{conv}\{\queries(\dsrow): \dsrow \in
\uni\}\), where
\(\queries(\dsrow) = (q_1(\dsrow), \ldots, q_k(\dsrow))\).
\(\spoly_\queries\) contains all possible values of \(\queries(\ds)\)
for all possible dataset sizes. A theorem  of Liaw, Mehrabian, Plan, and
Vershynin~\cite{JL-additive} guarantees that multiplication by a
suitable random matrix \(\jl \in \R^{\ell \times \qsize}\) with
\(\jldim \lesssim \frac{\log \usize}{\alpha^2}\) preserves all
distances between points in \(\spoly_\queries\) up to
\(\pm \alpha \sqrt{\qsize}\). Therefore, to achieve error
\(O(\alpha)\), it suffices to answer the projected queries \(\jl
\queries(\ds)\) with an answer vector \(\widetilde{Y} \in \jl \spoly_\queries\) with error
\(O(\alpha)\), and output a vector \(\widehat{Y} \in \spoly_\queries\)
such that \(\jl \widehat{Y} = \widetilde{Y}\). This works because, by
the result of Liaw et al., the distance between
\(\widehat{Y}\) and \(\queries(\ds)\) is approximately the same as the
distance between \(\widetilde{Y}\) and \(\jl\queries(\ds)\). The
benefit of using the random projection (i.e., multiplication by
\(\jl\)) is that now we only have to
answer the \(\jldim\)-dimensional workload \(\jl\queries\), for which
we can use the  K-norm mechanism, giving us error
\(\alpha\) with sample complexity
\(
\frac{\sqrt{\jldim \log \usize}}{\eps \alpha}
\lesssim
\frac{\log \usize}{\eps \alpha^2},
\)
as was our goal. The only caveat is that the K-norm mechanism does not
necessarily produce answers in \(\jl \spoly_\queries\), and
multiplication by \(\jl\) is only promised to preserve distances in \(\spoly_\queries\).
This problem has a simple fix: a least squares projection of the
noisy answers back onto  \(\jl \spoly_\queries\) does not increase the
error.

We note that the Johnson-Lindenstrauss lemma has been used before in
differential privacy in other contexts, e.g.~\cite{BlockiBDS12,KenthapadiKMM13,Stausholm21}. To the best of our knowledge,
the application to query released we present here is new.

A more careful analysis of this JL mechanism allows
giving the sample complexity bounds in terms of a natural geometric
property of the queries.
Intuitively, we expect that the sample
complexity required to compute \(\queries\) with differential privacy
scales with the ``size'' of \(\spoly_\queries\). In many situations,
the appropriate measure of size appears to be the Gaussian mean width,
defined for a set \(S \subseteq \R^\qsize\) by
\[
w(S) = \E\sup_{y \in S} \ip{y}{G},
\]
where \(G\) is a standard \(\qsize\)-dimensional Gaussian vector. The
novel bounds in Theorem~\ref{thm:main-worst-case} are implied and
refined by the next theorem. The theorem also refines the
computational complexity guarantee in
Theorem~\ref{thm:main-worst-case}, by showing that our mechanism runs
in  time polynomial in \(\qsize,\dsize,\frac{1}{\eps}\), and in the
running time of a separation oracle for \(\spoly_\queries\). Such an
oracle can be implemented in time polynomial in \(\qsize\) and
\(\usize\) using linear programming.
\begin{theorem}\label{thm:meanw}
For any workload of \(\qsize\) queries \(\queries\) over a universe of size \(\usize\), and any \(0 \le \alpha \le \frac{\diam(\spoly_\queries)}{\sqrt{\qsize}}\),
there exists an \(\eps\)-differentially private mechanism \(\mech\) with sample complexity
\[
\sc^{\ell_2}(\mech,\queries,\alpha) 
\lesssim
\min\left\{
  \frac{w(\spoly_\queries)}{\alpha\eps},
  \frac{w(\spoly_\queries)^2}{{\qsize}\eps \alpha^2}
\right\}
\]
Moreover, given a separation oracle for \(\spoly_\queries\), this
sample complexity is achieved by a mechanism running in time
polynomial in \(\qsize,\dsize,\frac{1}{\eps}\), and in the running
time of the separation oracle.
\end{theorem}
For some workloads, Theorem~\ref{thm:meanw} can give tighter bounds
than Theorem~\ref{thm:main-worst-case}. Moreover, it can apply in
situations when the universe \(\uni\) is not finite. We will give one
such example soon. To prove the theorem, we first show, using
Urysohn's inequality tying Gaussian width and volume, that the error
of the K-norm mechanism can be controlled in terms of the Gaussian
width of \(\spoly_\queries\). This gives the first term in the sample
complexity bound. We then apply the
JL mechanism with this noise bound, together with
the fact that multiplication by the random matrix \(\jl\) does not
asymptotically increase Gaussian width in expectation, which we prove.

Theorems~\ref{thm:main-worst-case}~and~\ref{thm:meanw} give improved,
and in fact tight sample complexity bounds for natural and important
workloads, e.g., for workloads of constant width marginal queries, and
workloads deriving from covariance estimation. Let \(\uni =
\{0,1\}^d\), and, for \(S \subseteq [d]\), define the query \(q_S(x) =
\prod_{i \in S} x_i\), i.e., the conjunction of the bits indexed by
\(S\). The workload of width-\(w\) marginals over \(d\)-dimensional
boolean data is defined as \(\queries_{w,d} = \{q_S: S\subseteq [d],
|S| = w\}\). Marginals capture many public releases of statistics in
the form of tables, and privately approximating marginal queries is
one of the better studied problems in differential
privacy~\cite{BarakCDKMT07,HardtRS12,CheraghchiKKL12,ChandrasekaranTUW14,ThalerUV12,conjunctions}. Using Theorem~\ref{thm:main-worst-case}, we
immediately get new
sample complexity bounds for computing marginals with pure
differential privacy. In particular, \eqref{eq:sc-opt} implies that, for any constant integer \(w\),
\begin{equation}
    \label{eq:sc-marg-opt}
    \sc_\eps^{\ell_2}(\queries_{w,d},\alpha)
    \lesssim
    \min\left\{
    \frac{d^w}{\eps \alpha},
    \frac{d^{(w+1)/2}}{\eps \alpha},
    \frac{d}{\eps \alpha^2}
    \right\}.
\end{equation}
The
third term in the sample complexity bound \eqref{eq:sc-marg-opt} is
new to this paper. We also prove a lower bound showing that
\eqref{eq:sc-marg-opt} is tight.

As another application, we also study the problem of estimating the
covariance of a bounded distribution in the Frobenius norm. The next
theorem crucially uses the refined Theorem~\ref{thm:meanw}, since
covariance matrices are naturally modeled using
statistical queries on an infinite domain. 
\begin{theorem}\label{thm:cov}
  There exists an \(\eps\)-differentially private mechanism that,
  given \(\dsize\) samples from a distribution \(\mu\) on \(B_2^d\)
  with covariance matrix \(\Sigma\), and outputs a matrix
  \(\widehat{\Sigma}\) such that, for any \(\alpha \in (0,1)\),
  \[
    \E \|\widehat{\Sigma} - \Sigma\|_F \le \alpha
  \]
  as long as
  \[
    \dsize \ge C\left( \min\left\{\frac{d^{1.5}}{\alpha\eps},
          \frac{d}{\alpha^2\eps}\right\}
      + \frac{1}{\alpha^2}
    \right)
  \]
  for a constant \(C\).
  Moreover, the mechanism runs in time polynomial in
  \(d,\dsize,\frac{1}{\eps}\). 
\end{theorem}
Theorem~\ref{thm:cov} improves on the recent work of Dong, Liang, and
Yi~\cite{DongLY-cov}, who give an algorithm with sample complexity on
the order of \(\frac{d^{1.5}}{\alpha^2\eps}\). Moreover, it shows that
the optimal sample complexity bound \eqref{eq:sc-marg-opt} can be achieved in
time polynomial in \(d\) when \(w=2\), since 2-way marginals can be
modeled as a special case of covariance estimation. Such efficient
algorithms are not known for \(w\ge 3\).

So far, we argued that the JL mechanism is optimal for worst-case
query workloads. We can in fact show that in the constant error regime
the sample complexity of a variant of the JL mechanism is also tight
with respect to the optimal sample complexity for the given
workload. This gives a partial answer to question 2.~from the
beginning of the introduction for pure differential privacy. The next
theorem, capturing this optimality guarantee, follows from combining
Theorem~\ref{thm:meanw} with the methods of Blasiok, Bun, Nikolov, and
Steinke~\cite{cdp}. 

\begin{theorem}\label{thm:opt}
For any workload \(\queries\) of \(\qsize\) queries over a universe of
size \(\usize\), and any \(0\le \alpha \le
\frac{\diam(\spoly_\queries)}{\sqrt{\qsize}}\), there exists an
\(\eps\)-differentially private mechanism \(\mech\) with running time
polynomial in \(\qsize, \usize, \dsize, \frac{1}{\eps}\), and with sample complexity
\[
\sc^{\ell_2}(\mech,\queries,\alpha)
\lesssim
\frac{\diam(\spoly_\queries)}{\alpha \sqrt{\qsize}}
\log\left(\frac{\diam(\spoly_\queries)}{\alpha \sqrt{\qsize}}\right)^2
\sc_\eps^{\ell_2}(\queries,\alpha/8).
\]
In particular, when \(\qsize\) is a workload of bounded queries,
\[
\sc^{\ell_2}(\mech,\queries,\alpha)
\lesssim
\frac{1}{\alpha}
\log\left(\frac{1}{\alpha}\right)^2
\sc_\eps^{\ell_2}(\queries,\alpha/8).
\]
\end{theorem}
In the case of bounded queries, Theorem~\ref{thm:opt} implies that, as long as \(\alpha\) is a
constant, no other mechanism can achieve much better sample complexity
and more than a constant factor better error at the same time. A similar
result was proved previously for pure differential privacy by Roth for
the SmallDB mechanism~\cite{RothNotes} (see
also~\cite{BunS16,Vadhan17}). Our algorithm runs in time polynomial in
\(\usize\), whereas Roth's algorithm needs to enumerate
over a set of size \(\usize^{1/\alpha^2}\). A similar result is also
known for concentrated differential privacy~\cite{cdp}.

The main idea in the proof of Theorem~\ref{thm:opt} is to run the JL
mechanism on an approximation \(K_\alpha \subseteq \spoly_\queries\)
of \(\spoly_\queries\). We construct \(K_\alpha\) by taking the convex
hull of a net of the extreme points of \(\spoly_\queries\). This
guarantees that any point in \(\spoly_\queries\) is within distance
\(\ll \alpha\sqrt{\qsize}\) from some point in \(K_\alpha\). At the
same time, in some cases, the Gaussian width of \(K_\alpha\) can be a
lot smaller than that of \(\spoly_\queries\), giving us better error
via Theorem~\ref{thm:meanw}.

\section{Preliminaries}

In this section we record some notation and preliminaries that we need
in the rest of the paper.

In terms of general notation, we use \(\|x\|_p = (\sum_{i = 1}^\qsize
|x_i|^p)^{1/p}\) to denote the \(\ell_p\) norm of a vector \(x \in
\R^\qsize\). We use \(B_p^\qsize = \{x \in \R^\qsize: \|x\|_p \le
1\}\) to denote the unit \(\ell_p\) ball in \(\R^\qsize\). We use
\(\|A\|_F = \sqrt{\tr(AA^T)}\) for the Frobenius norm of a matrix
\(A\), which is just its \(\ell_2\) norm when treated as a vector.

For a set
\(S \subseteq \R^\qsize\) and a real number \(t\), we use \(tS = \{tx: x\in S\}\) to denote
scaling. The Minkowski sum of two sets \(S, T \subseteq \R^\qsize\)
is denoted \(S + T = \{x+y: x\in S, y \in T\}\), and the special case
when one of the sets is a singleton is denoted \(x + S\) rather than
the more cumbersome \(\{x\} + S\).

We use \(A \otimes B\) to denote the Kronecker product of two matrices
\(A \in \R^{m \times n}\) and \(B \in \R^{k \times \ell}\), which is a
\(mk \times n\ell\) matrix whose rows and columns are indexed,
respectively, by pairs of rows and columns of \(A\) and \(B\). The
entry \((i,p),(j,q)\) of \(A\otimes B\) is equal to
\(a_{i,j}b_{p,q}\). We also use \(A^{\otimes w}\) for the \(w\)-fold
Kronecker product of \(A\) with itself. An important property of
Kronecker products is that \((A\otimes B)(x\otimes y) =
(Ax)\otimes (By)\) for matrices \(A,B\) and vectors \(x,y\) of
compatible dimensions. 

\subsection{Statistical Queries and Differential Privacy}

Let us recall the main notation on statistical queries from the Introduction.
We consider datasets $\ds \in \uni^\dsize$, and query workloads $\queries = \{\query_1, \ldots, \query_\qsize\}$, where each $\query_i$ is defined by a function $\query_i:\uni \to \R$, and, overloading notation, the query itself is defined by
\[
\query_i(\ds) = \frac{1}{\dsize}\sum_{\dsrow \in \ds}{\query_i(\dsrow)}.
\]
We say the queries in the workload are bounded if, for each \(i\in [\qsize]\) and each \(\dsrow\in \uni\), \(|\query_i(\dsrow)| \le 1\).
Again overloading notation, we write \(\queries(\ds)\) for the vector \((\query_1(\ds), \ldots, \query_\qsize(\ds))\), and, for any \(\dsrow \in \uni\), we write $\queries(\dsrow)$ for \((\query_1(\dsrow), \ldots, \query_\qsize(\dsrow))\). Then, approximating \(\queries(\ds)\) is equivalent to approximating the mean of \(\qsize\)-dimensional vectors:
\[
\queries(\ds) = \frac{1}{\dsize}\sum_{\dsrow \in \ds}{\queries(\dsrow)}.
\]
We define the sensitivity polytope \(\spoly_\queries =
\mathrm{conv}\{\queries(\dsrow): \dsrow \in \uni\}\), containing all
possible values of \(\queries(\ds)\) for all possible dataset
sizes. 

We define the \(\ell_p\) error of a mechanism \(\mech\) on a workload \(\queries\) and datasets of size \(n\) as 
\[
\err^{\ell_p}(\mech,\queries,n) = \max_{\ds \in \uni^\dsize} \E\left[ \frac{1}{\sqrt{\qsize}}\|\mech(\ds) - \queries(\ds)\|_p\right].
\]
Above, we assume that on input \(\ds\), the mechanism \(\mech\) outputs a random \(\qsize\)-dimensional vector \(\mech(\ds)\). The sample complexity of the mechanism on queries \(\queries\) is the smallest dataset size for which the mechanism can achieve error at most \(\alpha\), denoted
\[
\sc^{\ell_p}(\mech, \queries, \alpha) = \inf\{n: \err(\mech,\queries,n) \le \alpha\}.
\]
The optimal sample complexity of a workload \(\queries\) under
\(\eps\)-differential privacy is given by \[\sc^{\ell_p}_\eps(\queries,\alpha) = \inf\{\sc^{\ell_p}(\mech, \queries,\alpha): \mech \text{ is } \eps \text{-differentially private}\}.\]

In the rest of the paper we use \(\err(\mech, \queries,n)\),
\(\sc(\mech, \queries, \alpha)\), and \(\sc_\eps(\queries,
\alpha)\) without superscripts to denote,
respectively, the \(\ell_2\) error and the corresponding sample
complexity functions.

In the rest of the paper we call an \(\eps\)-differentially private mechanism for
answering a workload \(\queries\) of \(\qsize\) queries on datasets of
size \(\dsize\) efficient if it runs in polynomial time in
\(\qsize\),\(\dsize\), and \(\frac1\eps\) when given an evaluation
oracle for each query in \(\queries\), and a separation oracle for the
sensitivity polytope \(\spoly_\queries\).

We also mention two important properties of differential privacy:
composition and invariance under post-processing, both captured by the
following lemma. See the monograph of Dwork and Roth~\cite{DR14-monograph} for a proof.
\begin{lemma}\label{lm:composition}
  If \(\mech_1(\ds)\) is an \(\eps_1\)-differentially private mechanism,
  and \(\mech_2(\ds,y)\) is \(\eps_2\)-differentially private with respect to
  \(\ds\) for each \(y\) in the range of \(\mech_1\), then the
  composition \(\mech(\ds) = \mech_2(\ds,\mech_1(\ds))\) is \((\eps_1
  + \eps_2)\)-differentially private. In particular, for any
  \(\eps\)-differentially private mechanism \(\mech\), and any,
  potentially randomized, function \(f\) on the range of \(\mech\),
  \(f(\mech(\ds))\) is \(\eps\)-differentially private.
\end{lemma}

\subsection{Packing Lower Bounds}

The following packing lower bound for statistical queries is likely
folklore, and was used in the work of Hardt and
Talwar~\cite{HardtT10} (see also Theorem 2.1 in \cite{De12}).

\begin{lemma}\label{lm:packing-general}
  Suppose that there exists a workload \(\queries\) of \(\qsize\)
  queries, and  datasets \(\ds_1, \ldots, \ds_M \in \uni^\dsize\) such that
  for any \(i \neq j\), \(i,j \in [M]\),
  \(
  \frac{1}{\sqrt{\qsize}} \|\queries(\ds_i) - \queries(\ds_j)\|_2
  >
  \alpha,
  \)
  and \(\ds_i\) and \(\ds_j\) differ in at most \(\Delta\)
  elements. Then, as long as \(\Delta \le \frac{\log(M/2)}{\eps}\),
  we have
  \(
  \sc_\eps(\queries,\alpha/2) \ge \dsize.
  \)
\end{lemma}

The next lemma is a direct consequence of
Lemma~\ref{lm:packing-general}, and is often, but not always how the
packing bound is used. In it, we use the notion of a separation
number: for a compact set \(K\subseteq \R^\qsize\) and a centrally symmetric convex set
\(L\), the separation number \(\sep(K,L)\) is the maximum size of a
finite set \(S \subseteq K\) such that for all distinct \(y,y' \in S\) we
have \(y \not \in y' + L\). 

\begin{lemma}\label{lm:packing-singleton}
  Let \(\queries\) be a workload of \(\qsize\)
  queries, and let \(S_\queries = \{\queries(\elem): \elem \in
  \uni\}\) be the set of vectors of query answers on each universe
  element. Then for any \(\alpha \in (0,t]\),
  \(
  \sc_\eps(\queries,\alpha/2) \gtrsim
  \frac{t\log (\sep(S_\queries, t \sqrt{\qsize} B_2^\qsize)/2)}{\alpha \eps}.
  \)
\end{lemma}
\begin{proof}
  Let \(\elem_1, \ldots, \elem_M\) achieve \(\sep(S_\queries,
  t \sqrt{\qsize} B_2^\qsize)\). This means that, for any two distinct
  \(\elem_i, \elem_j\), \(\frac{1}{\sqrt{\qsize}}\|\queries(\elem_i) -
  \queries(\elem_j)\|_2 > t\).
  Let \(\dsize = \frac{t\log(M/2)}{\alpha \eps}\).
  We pick an arbitrary \(\elem_0\), and define the dataset \(\ds_i\)
  to contain \(\frac{\alpha}{t} \dsize\) copies of \(\elem_i\) and
  \((1-\frac{\alpha}{t})\dsize\) copies of \(\elem_0\). Then \(\Delta
  = \frac{\log(M/2)}{\eps}\), and, for any \(i \neq j\),
  \(\frac{1}{\sqrt{\qsize}}\|\queries(\ds_i) - \queries(\ds_j)\|_2 >
  \alpha\). We can now apply Lemma~\ref{lm:packing-general} to finish
  the proof.
\end{proof}

\subsection{A Johnson-Lindenstrauss Lemma with Additive Error}

Our algorithms rely on a variant of the Johnson-Lindenstrauss (JL) lemma
with additive error, due to Liaw, Mehrabian, Plan, and
Vershynin~\cite{JL-additive}. Their result shows that the pairwise
distances between elements in any bounded (but potentially infinite) set are preserved under
suitable random projection onto a lower dimensional subspace, up to
an additive approximation. Before stating this result precisely, we recall the
definitions of subgaussian random variables and the \(\psi_2\) Orlicz
norm. 

The \(\psi_2\) norm of, respectively, a real-valued random variable,
and an \(\jldim\)-dimensional random vector are defined by
\[
  \|A\|_{\psi_2} = \inf\{t: \E \psi_2(A/z) \le 1\};
  \ \ \ \
  \|V\|_{\psi_2} = \sup_{\theta \in \R^\ell: \|\theta\|_2 = 1} \|\ip{V}{\theta}\|_{\psi_2},
\]
where \(\psi_2(a) = e^{a^2} - 1\).
The random variable is \emph{subgaussian} if its \(\psi_2\) norm is finite.
Background on subgaussian random variables, including other equivalent
definitions, their properties, and examples, can be found in
Vershynin's book~\cite{Vershynin-HDP}. Here we mention two examples:
\begin{itemize}
\item A Gaussian random vector \(G\sim N(0,\sigma^2 I)\) is
  subgaussian with \(\|G\|_{\psi_2} \lesssim \sigma\);
\item A scaled symmetric Bernoulli random vector \(B\), where \(B_i\)
  is chosen independently and uniformly in \(\{-\sigma, +\sigma\}\), is
  subgaussian with \(\|B\|_{\psi_2} \lesssim \sigma\).
\end{itemize}
Next we define the class of random
projections for which the additive JL lemma is known to hold.

\begin{definition}\label{def:jl-matrix}
  A random matrix \(\jl \in \R^{\jldim \times \qsize}\) is a
  \emph{JL-matrix with parameter \(C\)} if it has independent rows
  \(\jl_1, \ldots, \jl_i\), each satisfying
\begin{align*}  \label{eq:jl-assumptions}
  \E \jl_i \jl_i^\top = \frac{1}{\jldim} I;\ \ \ \ 
  \|\jl_i\|_{\psi_2} \le \frac{C}{\sqrt{\jldim}}.
\end{align*}
\end{definition}
Examples of JL-matrices with parameter \(C \lesssim 1\) include matrices
with IID entries, each entry picked either from
\(N(0,\frac{1}{\jldim})\), or uniformly from
\(\left\{-\frac{1}{\sqrt{\jldim}}, \frac{1}{\sqrt{\jldim}}\right\}\).


We can now state the additive JL lemma that we use.
\begin{theorem}[\cite{JL-additive}]\label{thm:jl}
Let \(\jl \in \R^{\jldim \times \qsize}\) be a JL matrix with
parameter \(C\). Then, for any bounded set \(S \subseteq \R^\qsize\), we have that
\[
\E\sup_{y,y'\in S}\left|\|\jl y - \jl y'\|_2  - \|y - y'\|_2
\right|\lesssim \frac{C^2 w(S)}{\sqrt{\ell}}.
\]
\end{theorem}


In the remainder of this section, we derive a useful consequence of
Theorem~\ref{thm:jl}. We prove that, for any bounded
set \(S\), the expected Gaussian width of \(\jl S\) is bounded by the
Gaussian width of \(S\). To this end,  we need a lemma about
Gaussian processes, stating that contracting distances does not
increase the Gaussian mean width. For a short proof of a slightly
weaker lemma (which would suffice for our application), see Corollary 3.14 in Ledoux and Talagrand's book~\cite{LedTal}.

\begin{lemma}[Theorem 3.15~\cite{LedTal}]\label{lm:gauss-contraction}
Let \(S \subseteq \R^\qsize\) be a bounded set, and let \(f:\R^\qsize \to \R^\ell\) be a function such that
\[
\forall y,z \in S:
\|f(y) - f(z)\|_2 \le C \|y-z\|_2.
\]
Then \(w(f(S)) \le C w(S)\).
\end{lemma}

We use Lemma~\ref{lm:gauss-contraction} to prove a similar result when we allow an additive distortion in pairwise distances.

\begin{lemma}\label{lm:gauss-additive}
Let \(S \subseteq \R^\qsize\) be a bounded set, and let \(f:\R^\qsize \to \R^\ell\) be a function such that
\[
\forall y,y' \in S:
\|f(y) - f(y')\|_2 \le \|y-y'\|_2 + \alpha
\]
for some \(\alpha \ge 0\). Then \(w(f(S)) \lesssim w(S) + \alpha\sqrt{\ell}\).
\end{lemma}
\begin{proof}
Let \(T\subseteq S\) be an inclusion maximal set such that \(f(T)\) is a \(2\alpha\)-separated subset of \(f(S)\). I.e., for any two distinct points \(y,y' \in T\), we have \(\|f(y)-f(y')\|_2 > 2\alpha\), and, moreover, we cannot add any point of \(S\) to \(T\) while still satisfying this property. This means that for any \(z \in f(S)\) there is some \(z' \in f(T)\) such that \(\|z - z'\|_2 \le 2\alpha\); let us fix one such \(z'\in f(T)\) for each \(z \in f(S)\) and denote it \(\pi(z)\). Taking \(G\) to be a standard Gaussian random vector in \(\R^\ell\), we have
\begin{align}
w(f(S)) &= \E \sup_{z \in f(S)} \ip{z}{G}\notag\\
&\le \E \sup_{z \in f(S)} \ip{\pi(z)}{G} + \E \sup_{z\in f(S)} \ip{z - \pi(z)}{G}\notag\\
&\le \E \sup_{z' \in f(T)} \ip{z'}{G} + \sup_{z\in f(S)} \|z -\pi(z)\|_2 \E\|G\|_2\label{eq:CS}\\
&\le w(f(T)) + 2\alpha \sqrt{\ell}. \label{eq:gauss-net}
\end{align}
Above, inequality \eqref{eq:CS} follows by the Cauchy-Schwarz
inequality. At the same time, we can use
Lemma~\ref{lm:gauss-contraction} with \(T\) and \(f(T)\), because
\(f(T)\) is well-separated. Indeed, for any two distinct \(y,y'\in T\), by assumption, and because \(f(T)\) is \(2\alpha\)-separated, we have 
\[
\|y-y'\|_2 \ge \|f(y)-f(y')\|_2 - \alpha > \frac{1}{2} \|f(y)-f(y')\|_2.
\]
Therefore, by Lemma~\ref{lm:gauss-contraction}, 
\[
w(f(T)) \le 2 w(T) \le 2w(S),
\]
with the final inequality following since \(T \subseteq S\). Combining this bound with \eqref{eq:gauss-net} proves the lemma.
\end{proof}

\begin{lemma}\label{lm:jl-gauss}
Let \(\jl\in \R^{\ell\times \qsize}\) be a JL-matrix with parameter \(C\). Then, for any bounded set \(S\subseteq \R^\qsize\), we have
\[
\E w(\jl S) \lesssim C^2 w(S).
\]
\end{lemma}
\begin{proof}
Lemma~\ref{lm:gauss-additive} implies that pointwise for each \(\jl \in \R^{\ell\times \qsize}\) we have
\[
w(\jl K) \lesssim w(K) + \left(\sup_{y,y'\in K} \left|\|\jl y - \jl y'\|_2 - \|y - y'\|_2\right|\right)\sqrt{\ell}.
\]
Taking expectations of both sides over \(\jl\) and using Theorem~\ref{thm:jl}, we get
\[
\E w(\jl K) \lesssim w(K) + \frac{C^2 w(K)}{\sqrt{\ell}} \sqrt{\ell}
=(1 + C^2)w(K).
\]
This finishes the proof, since \(C \gtrsim 1\) for
Definition~\ref{def:jl-matrix} to hold (as follows, e.g., from item
(ii) of Proposition 2.5.2 in~\cite{Vershynin-HDP}).
\end{proof}

\section{The JL Mechanism}


We consider algorithms that follow a general template, presented in
Algorithm~\ref{alg:template}. Our main lemma for the analysis of this algorithm is the following.

\begin{algorithm}
\begin{algorithmic}

    \State Pick a JL-matrix \(\jl \in \R^{\jldim \times \qsize}\) with
    parameter \(C \lesssim 1\).
   
    \State \(\widetilde{Y} \leftarrow \jl \queries(\ds) + \frac{1}{\dsize}Z\)
    \Comment{\(Z\in \R^k\) is additive noise that ensures \(\widetilde{Y}\) is \(\varepsilon\)-differentially private.}
    
    \State Output
    \(
    \widehat{Y} \in \arg \min \left\{
    \|\widetilde{Y} - \jl y\|_2: y \in  \spoly_\queries
    \right\}\)
  \end{algorithmic}
  \caption{The JL-Release mechanism template}
  \label{alg:template}
\end{algorithm}

\begin{lemma}\label{lm:template}
Let \(\mech_Z\) be the mechanism described in
Algorithm~\ref{alg:template}, where \(\jl\) is a JL-matrix with
parameter \(C\). Then there exists an absolute constant \(C'>0\) such
that 
\begin{equation}\label{eq:template-error}
\err(\mech_Z,\queries,n) = \E
\frac{1}{\sqrt{\qsize}}\|\widehat{Y} - \queries(\ds)\|_2 \le \frac{C'C^2
  w(\spoly_\queries)}{\sqrt{\qsize \jldim}} + \frac{1}{\sqrt{\qsize}\dsize}\E\|Z\|_2.
\end{equation}
Moreover, if the mechanism that outputs \(\widetilde{Y} = \jl
\queries(\ds) + \frac{1}{\dsize}Z\) on input \(\ds\) satisfies
\(\eps\)-differential privacy (resp.~\((\eps,\delta)\)-differential
privacy, \(\rho\)-zCDP, \(f\)-DP) with respect to \(\ds\), then so
does \(\mech_Z\). \(\mech_Z\) is also efficient assuming that \(Z\)
can be sampled efficiently.
\end{lemma}
\begin{proof}
We first claim that 
\begin{equation}
    \label{eq:proj-err}
    \|\jl \widehat{Y} - \jl\queries(\ds)\|_2 \le \|\widetilde{Y} - \jl\queries(\ds)\|_2 = \frac{1}{\dsize}\|Z\|_2.
\end{equation}
This is a standard fact about \(\ell_2\) projection onto a convex
set. Indeed, by first-order optimality conditions, the set \(\{y\in
\R^\qsize: \ip{\jl \widehat{Y} - \widetilde{Y}}{\jl y-\jl \widehat{Y}} \ge 0\}\) contains all of \(\spoly_\queries\). Since \(\queries(\ds) \in \spoly_\queries\), this means that 
\[
 \ip{\widetilde{Y} - \jl \widehat{Y}}{\jl \widehat{Y} - \jl
   \queries(\ds)}
 =
 \ip{\jl \widehat{Y} - \widetilde{Y}}{\jl \queries(\ds)-\jl \widehat{Y}}
\ge 0.
\] We have
\begin{align*}
    \|\widetilde{Y} - \jl\queries(\ds)\|_2^2 
    &= \|\widetilde{Y} - \jl \widehat{Y} + \jl \widehat{Y}-\jl\queries(\ds)\|_2^2\\
    &= \|\widetilde{Y} - \jl \widehat{Y}\|_2^2 + 2 \ip{\widetilde{Y} - \jl \widehat{Y}}{\jl \widehat{Y}-\jl\queries(\ds)} + \|\jl \widehat{Y}-\jl\queries(\ds)\|_2^2\\
    &\ge \|\jl \widehat{Y}-\jl\queries(\ds)\|_2^2.
\end{align*}
This proves \eqref{eq:proj-err}. From \eqref{eq:proj-err} and Theorem~\ref{thm:jl}, we have that
\begin{align}
    \E \|\widehat{Y} - \queries(\ds)\|_2 &\le
    \E\left|\|\jl \widehat{Y}-\jl\queries(\ds)\|_2 - \|\widehat{Y} - \queries(\ds)\|_2\right| + \E\|\jl \widehat{Y}-\jl\queries(\ds)\|_2 \notag\\
    &\le \E \sup_{y,y' \in \spoly_\queries}\left|\|\jl y -\jl y'\|_2 - \|y - y'\|_2\right| + \frac{1}{\dsize}\E\|Z\|_2\label{eq:sup}\\
    &\le \frac{C'C^2
  w(\spoly_\queries)}{\sqrt{\ell}} + \frac{1}{\dsize}\E\|Z\|_2\label{eq:apply-jl}.
\end{align}
Above, in inequality \eqref{eq:sup} we use \eqref{eq:proj-err} and the
fact that \(\widehat{Y}\) and \(\queries(\ds)\) are both elements of
\(\spoly_\queries\), and in inequality \eqref{eq:apply-jl} we use Theorem~\ref{thm:jl}. The desired bound on error now follows by dividing through by \(\sqrt{\qsize}\).

The privacy guarantee follows directly from the post-processing
property of differential privacy
(Lemma~\ref{lm:composition}). Efficiency follows since 
\(\|\widetilde{Y} - \jl y\|_2\), being a convex function, can be minimized over \(y \in
\spoly_\queries\) in time polynomial in \(\qsize\) given a separation
oracle for \(\spoly_\queries\) (see, e.g., the Frank-Wolfe
implementation of this projection step in~\cite{conjunctions}).
\end{proof}

Note that, if \(\jl\) is a JL-matrix with parameter \(C \lesssim 1\),
then, in order make the first term in the  error bound
\eqref{eq:template-error} at  most \(\frac{\alpha}{2}\), it is enough
to choose \(\jldim \gg \frac{w(\spoly_\queries)^2}{\qsize \alpha^2}\).

Finally, let us remark when \(\jl\) is deterministically chosen to be
the identity matrix, Algorithm~\ref{alg:template} reduces to the
Projection Mechanism from~\cite{NTZ}, i.e., we just add noise in order
to preserve differential privacy, and perform a least squares
projection of the noisy query answers onto \(\spoly_\queries\). In the
case of \((\eps, \delta)\)-differential privacy and bounded queries on
a universe of size \(\usize\),
the Projection Mechanism, instantiated with Gaussian noise, achieves
sample complexity (with respect to \(\ell_2\) error) on the order of\(\frac{\sqrt{\log \usize \log
    1/\delta}}{\alpha^2\eps}\), which is also optimal for worst-case
queries~\cite{NTZ,BunUV14}. The same guarantee also follows directly
from Lemma~\ref{lm:template} when \(Z\) is chosen to be Gaussian
noise. Nevertheless, it is interesting that, in that case, the
dimension reduction step appears to be unnecessary. We do not know if it is
necessary for pure differential privacy, i.e., if the Projection
Mechanism can recover optimal sample complexity with noise achieving
pure differential privacy. Nikolov, Talwar, and Zhang analyzed the
Projection Mechanism with K-norm noise~\cite{NTZ}, but their bound is
suboptimal. The main hurdle to getting an optimal bound using their
approach is that one-dimensional marginals of K-norm noise are not subgaussian.

\section{Additive Noise Distributions for Pure Privacy}

Our goal is to instantiate Algorithm~\ref{alg:template} with different
choices of the noise \(Z\). To this end, we recall and refine the
\(K\)-norm mechanisms, introduced by Hardt and Talwar~\cite{HardtT10}.

Let \(K\) be a convex body (convex bounded set with non-empty
interior) in \(\R^\qsize\), which is symmetric around the origin,
i.e., \(K = -K\). Then \(K\) is the unit ball of the norm
\(\|\cdot\|_K\), defined by \(\|y\|_K = \inf\{t \ge 0: y \in
tK\}\). The \(K\)-norm distribution \(\dist_{K,\eps}\) is a
probability distribution on \(\R^\qsize\) with probability density
function
\[
f_K(y) = \frac{\eps^{k}e^{\eps \|y\|_K}}{\Gamma(k+1)\vol(K)}.
\]

We have the following result, which is a very slight strengthening of
a similar result (Theorem 4.3) in~\cite{HardtT10}. Hardt and Talwar
stated it
for the special case \(K = \spoly_\queries\),  but the same proof
establishes the lemma stated below.
\begin{lemma}[\cite{HardtT10}]\label{lm:Knorm-priv}
If \(\spoly_\queries \subseteq K\), then the \(K\)-norm mechanism,
which outputs \(\queries(\ds) + \frac{1}{n} Z\) for \(Z \sim
\dist_{K,\eps}\) on input \(\ds \in \uni^\dsize\), satisfies
\(\eps\)-differential privacy. Moreover, 
\begin{equation}
    \label{eq:Knorm-noise}
    \E \|Z\|_2 \le \frac{\qsize+1}{\eps} \E \|U\|_2
    \ \ \ \ \ \ 
    \E \|Z\|_2^2 \le \frac{(\qsize+2)(\qsize+1)}{\eps^2} \E \|U\|_2^2,
  \end{equation}
where \(U\) is a random vector distributed uniformly in
\(K\). Therefore, the error of the \(K\)-norm mechanism is bounded by
\(\frac{\qsize+1}{\sqrt{\qsize}n\eps} \E \|U\|_2
\lesssim
\frac{\sqrt{\qsize}}{n\eps} \E \|U\|_2.
\)
\end{lemma}

The next lemma follows easily from Lemma~\ref{lm:Knorm-priv}.
\begin{lemma}\label{lm:ball-noise}
For any workload \(\queries\) of \(\qsize\) queries, and \(Z\sim
\dist_{\rad(\spoly_\queries)B_2^\qsize, \eps}\), the mechanism
\(\mech\) that outputs \(\queries(\ds) + \frac{1}{\dsize} Z\) is
efficient, \(\eps\)-differentially private and has error
\[
\err(\mech, \queries, \dsize) = \frac{1}{\sqrt{\qsize}\dsize}\E\|Z\|_2 
\lesssim \frac{\sqrt{\qsize}}{n\eps}\rad(\spoly_\queries).
\]
In particular, if \(\queries\) consists of bounded queries, the
mechanism has error
\(
\err(\mech, \queries, \dsize) \lesssim \frac{\qsize}{n\eps}.
\)
\end{lemma}

While the error in Lemma~\ref{lm:ball-noise} is optimal in some cases,
it can also be far from optimal when the size of the data universe
\(\uni\) is sub-exponential in dimension. In the remainder of this section, we describe a noise
mechanism with better error in this latter setting. The mechanism
closely follows the one given by Hardt and Talwar, but with a
different choice of how to spend the privacy budget, in order to
remove unnecessary logarithmic factors. The error bound and efficiency
guarantees of this optimized mechanism are given by the following lemma.
\begin{lemma}\label{lm:noise-meanw}
For any workload \(\queries\) of \(\qsize\) queries there exists a
noise random variable \(Z\) such that the mechanism \(\mech\) that
outputs \(\queries(\ds) + \frac{1}{\dsize} Z\) is efficient,
\(\eps\)-differentially private and has error
\[
\err(\mech, \queries, \dsize) = \frac{1}{\sqrt{\qsize}\dsize}\E\|Z\|_2 \lesssim 
\frac{w(\spoly_\queries)}{n\eps}.
\]
In particular, if \(\queries\) consists of bounded queries over a
universe of size \(\usize\), then the
mechanism has error
\(
\err(\mech, \queries, \dsize) \lesssim \frac{\sqrt{\qsize\log \usize}}{n\eps}.
\)
\end{lemma}

Towards proving Lemma~\ref{lm:noise-meanw}, let us recall the volume lower bound from Hardt and Talwar's paper. In formulating it, we use the notion of volume radius, defined for a convex body \(K \in \R^\qsize\) as 
\(
\vrad(K) = \left(\frac{\vol(K)}{\vol(B_2^\qsize)}\right)^{1/\qsize}.
\)
In other words, the volume radius of \(K\) is the radius of a Euclidean ball with the same volume as \(K\).
It will also be convenient to treat bounded convex sets that lie in a lower dimensional subspace of \(\R^\qsize\). Let us write \(\vol_W(\cdot)\) for volume (i.e., Lebesgue measure) restricted to a subspace \(W\) of \(\R^\qsize\). Then for a subspace \(W\) of dimension \(\ell\), and a convex body \(K \subseteq W\) (i.e., \(K\) is convex and bounded, and its interior relative to \(W\) is non-empty), we define the volume radius relative to \(W\) as 
\(
\vrad_W(K) = \left(\frac{\vol_W(K)}{\vol(B_2^{\ell})}\right)^{1/\ell}.
\)
We are now ready to formulate the volume lower bounds for a convex body \(K\) in \(\R^\qsize\), given by
\begin{equation*}
    \volLB(K,\ell) = \ell \sup_{W \in G_{\qsize,\ell}} \vrad_W(P_W K).
\end{equation*}
Above \(G_{\qsize,\ell}\) is the set of all \(\ell\)-dimensional
subspaces of \(\R^\qsize\), and \(P_W\) is the orthogonal projection
onto the subspace \(W\). Hardt and Talwar constructed an efficiently
samplable noise distribution whose error can be bounded in terms of
\(\max_{\ell} \volLB(K,\ell)\). Using the error guarantees of their
noise distribution as a black box, together with classical volume
estimates, we can prove an error bound weaker than the one in Lemma~\ref{lm:noise-meanw}
by a polylogarithmic factor. In order to give a tight bound, we
open up and optimize Hardt and Talwar's construction.

It will be convenient to use a monotonicity property of the volume
radii of projections of a convex body \(K\). This property follows
from the Alexandrov-Fenchel inequalities: see Chapter 9 of Pisier's
book~\cite{Pisier-volume} for a proof.
\begin{lemma}[\cite{Pisier-volume}]\label{lm:vnum-monotone}
For any symmetric convex body \(K\subseteq \R^\qsize\), and any \(1 < \ell \le \qsize\),
\[
\sup_{W \in G_{n,\ell}} \vrad_W(P_W K) \le \sup_{W \in G_{n,\ell-1}} \vrad_W(P_W K).
\]
\end{lemma}

The next lemma is the key technical tool in proving
Lemma~\ref{lm:noise-meanw}. The lemma restates the key claims in the
analysis of Hardt and Talwar's mechanism for non-isotropic
bodies~\cite[Section 7]{HardtT10}. While Hardt and Talwar's result was
conditional on the hyperplane conjecture,\footnote{A recent preprint
  by Klartag and Lehec claims a proof of the conjecture up to a
  polylogarithmic factor~\cite{KL-KLS}.} Bhaskara, Dadush,
Krishnaswamy, and Talwar showed how to make the result
unconditional~\cite{BhaskaraDKT12}. 
\begin{lemma}[\cite{HardtT10,BhaskaraDKT12}]\label{lm:decomp}
For any symmetric convex body \(K\subseteq \R^\qsize\), there exists subspaces \(W_0 = \R^\qsize, W_1, \ldots, W_m\) and symmetric bounded convex sets \(K_1, \ldots, K_m\) with the following properties:
\begin{itemize}
    \item \(\bigcup_{i=1}^m W_i\) spans \(\R^\qsize\);
    \item for all \(i\ge 2\), \(W_i\) is contained in the orthogonal complement of \(W_{i-1}\);
    \item for all \(i\ge 1\), \( \dim W_i = \left\lfloor\frac12 \dim W_{i-1}\right\rfloor\);
    \item for all \(i\ge 1\), \(K_i \subset W_i\) and \(K_i\) is a convex body inside \(W_i\), i.e., has a non-empty interior relative to \(W_i\);
    \item for all \(i\ge 1\), and for \(U_i\) distributed uniformly in \(K_i\), 
    \begin{equation}\label{eq:decomp-err}
        (\dim W_i)^2 \E \|U_i\|_2^2 \lesssim \volLB(K,\dim W_{i-1} - \dim W_i)^2;
    \end{equation}
    \item given a sampling oracle for \(K\), we can sample from each \(K_i\) in time polynomial in \(\qsize\).
\end{itemize}
\end{lemma}

Next we use Lemma~\ref{lm:decomp} to give an optimized version of
Hardt and Talwar's noise distribution.

In what follows, for a convex body \(K\) we use \(K^{\sym}\) for the symmetrized body \(K^\sym = \mathrm{conv}(K \cup -K)\).

\begin{lemma}\label{lm:Knorm-noise}
For any workload \(\queries\) of \(\qsize\) queries there exists a
noise random variable \(Z\) such that the mechanism \(\mech\) that
outputs \(\queries(\ds) + \frac{1}{\dsize} Z\) is efficient, \(\eps\)-differentially private, and has error
\[
\err(\mech, \queries, \dsize) = \frac{1}{\sqrt{\qsize}\dsize}\E\|Z\|_2 \lesssim 
\frac{1}{\sqrt{\qsize}n\eps}\left(\sum_{i = 1}^{\lfloor\log_2 \qsize \rfloor}{\volLB(\spoly^\sym_\queries,\lfloor 2^{-i} \qsize\rfloor)^{2/3}}\right)^{3/2}.
\]
\end{lemma}
\begin{proof}
Let \(W_1, \ldots, W_m\) and \(K_1, \ldots, K_m\) be as in
Lemma~\ref{lm:decomp}, used with \(K = \spoly^\sym_\queries\). Let
\(\qsize_i = \lfloor 2^{-i} \qsize\rfloor\). An easy induction
argument shows \(\dim W_i = \qsize_i\), and \(\dim W_m = \qsize_m = 1\).

We define \(Z\) to be distributed as \(Z_1 + \ldots + Z_m\), where
\(Z_i\) is a random variable supported on \(W_i\) distributed
independently according to \(\dist_{K_i,\eps_i}\), for a value of
\(\eps_i\) to be chosen shortly. (This choice of \(\eps_i\) is the
only place where our noise distribution differs from Hardt and
Talwar's.)  Since the subspaces \(W_1, \ldots, W_m\) are pairwise orthogonal, using Jensen's inequality we have
\begin{align}
\E\|Z\|_2 \le (\E\|Z\|_2^2)^{1/2}
&=
\left(\sum_{i = 1}^m \E\|Z_i\|_2^2\right)^{1/2}
\lesssim
\left(\sum_{i = 1}^m \frac{\qsize_i^2}{\eps_i^2}\E\|U_i\|_2^2\right)^{1/2}.\notag
\end{align}
The final inequality above follows from
\eqref{eq:Knorm-noise}. We choose \(\eps_1, \ldots, \eps_m\) so that \(\eps_1 + \ldots +
\eps_m = \eps\), which guarantees that the algorithm is
\(\eps\)-differentially private, as we prove below. The choice
that satisfies this constraint while minimizing the right hand side of
\eqref{eq:err-vollb} is 
\[
\eps_i = \frac{\eps (\qsize_i\E\|U_i\|_2)^{2/3}}{\sum_{i = 1}^m (\qsize_i\E\|U_i\|_2)^{2/3}},
\]
and gives
\begin{equation}\label{eq:err-Ui}
\E\|Z\|_2
\lesssim
\frac{1}{\eps}\left(\sum_{i = 1}^m (\qsize_i\E\|U_i\|_2)^{2/3} \right)^{3/2}.
\end{equation}
Note that this choice of \(\eps_1, \ldots, \eps_m\) can be computed
efficiently by estimating \(\E\|U_i\|_2\) through sampling. Then the
whole algorithm is efficient by Lemma~\ref{lm:decomp}.

Using \eqref{eq:decomp-err} 
we have
\[
\qsize_i^2 \E\|U_i\|_2^2 \lesssim \volLB(\spoly^\sym_\queries, \qsize_{i-1} - \qsize_i)^2.
\] 
Since \(\qsize_{i-1} - \qsize_i = \lceil \qsize_{i-1}/2\rceil\), and
\(\qsize_i \ge 1\), we have \(\qsize_i \le \qsize_{i-1} - \qsize_i \le 2\qsize_i\). By Lemma~\ref{lm:vnum-monotone} we then have
\begin{align*}
  \volLB(\spoly^\sym_\queries, \qsize_{i-1} - \qsize_i)
&=
(\qsize_{i-1} - \qsize_i) \sup_{W \in G_{\qsize,\qsize_{i-1} - \qsize_i}} \vrad_W(P_W K)^{1/(\qsize_{i-1} - \qsize_i)}\\
&=
2\qsize_i  \sup_{W \in G_{\qsize,\qsize_i}} \vrad_W(P_W K)^{1/\qsize_i}\\
&\le 2 \volLB(\spoly^\sym_\queries, \qsize_i).
\end{align*}
Substituting in \eqref{eq:err-Ui} gives
\begin{equation}\label{eq:err-vollb}
\E\|Z\|_2
\lesssim
\frac{1}{\eps}\left(\sum_{i = 1}^m \volLB(\spoly^\sym_\queries, \qsize_i)^{2/3} \right)^{3/2}.
\end{equation}
This implies the claimed bound on the error of the mechanism \(\mech\). To argue privacy, notice that \(\mech(\ds) = \queries(\ds) + \frac{1}{\dsize} Z\) is distributed identically to 
\[
\left(P_{W_1}\queries(\ds) + \frac{1}{\dsize}Z_1\right) + \ldots + \left(P_{W_m} \queries(\ds) + \frac{1}{\dsize}Z_m\right).
\]
Here we use the fact that \(P_{W_1} + \ldots + P_{W_m} = I\) since the
subspaces \(W_1, \ldots, W_m\) are pairwise orthogonal and span
\(\R^\qsize\)\cut{, and also the fact that \(P_{W_i} Z_i = Z_i\), since
\(Z_i\) is supported on \(W_i\)}. Let us define a mechanism
\(\mech_i(\ds) = P_{W_i}\queries(\ds) + \frac{1}{\dsize}Z_i\) for
every \(i\in [m]\). By Lemma~\ref{lm:Knorm-priv}, \(\mech_i\) is
\(\eps_i\)-differentially private. We can see \(\mech\) as a
postprocessing of the composition of \(\mech_1, \ldots, \mech_m\),
and, therefore, by Lemma~\ref{lm:composition}, \(\mech\) satisfies differential privacy with privacy parameter \(\eps_1 + \ldots + \eps_m = \eps\).
\end{proof}

A final tool we need in the proof of Lemma~\ref{lm:noise-meanw} is
Urysohn's inequality, which states that for any bounded set \(K \subseteq \R^\qsize\), we have 
\begin{equation}\label{eq:urysohn}
    \vrad(K) \le \frac{w(K)}{\sqrt{\qsize}}.
\end{equation}

We can now use the inequality~\eqref{eq:urysohn} and the fact that
mean width does not increase under projection to give a bound on the
volume lower bounds, and, thereby, give an upper bound on the
expression for the error in Lemma~\ref{lm:Knorm-noise}.

\begin{lemma}\label{lm:vollb-meanw}
For any convex body \(K\subseteq \R^\qsize\), and any \(1 \le \ell \le \qsize\), we have
\(
\volLB(K, \ell) \le \sqrt{\ell} w(K).
\)
\end{lemma}
\begin{proof}
The lemma follows directly from~\eqref{eq:urysohn}, and the
observation that \(w(PK) \le w(K)\) for any orthogonal projection
\(P\). The latter fact is well-known, and follows, for example, from Lemma~\ref{lm:gauss-contraction}.
\end{proof}

\begin{proof}[Proof of Lemma~\ref{lm:noise-meanw}]
For a technical reason, which we explain later, we need that \(0\in\spoly_\queries\). Since this is not necessarily the case, we define an auxiliary workload \(\widetilde{\queries}\). Let \(\elem_0 \in \uni\) be arbitrary, and define \(\widetilde{\queries}\) by \(\widetilde{\queries}(\elem) = \queries(\elem) - \queries(\elem_0)\) for all elements \(\elem \in \uni\). Equivalently, we have \(\widetilde{\queries}(\ds) = \queries(\ds) - \queries(\elem_0)\) for all datasets \(\ds \in \uni^\dsize\). Since, by this definition, \(\widetilde{\queries}(\elem_0)= 0\), we have \(0\in \spoly_{\widetilde{\queries}} = \spoly_\queries - \queries(\elem_0)\). 

Let us now take \(Z\) to be the noise distribution guaranteed by Lemma~\ref{lm:Knorm-noise} used with the workload \(\widetilde{\queries}\). By Lemma~\ref{lm:Knorm-noise}, the mechanism \(\widetilde{\mech}\) that outputs \(\widetilde{\queries}(\ds) + \frac{1}{\dsize} {Z}\) is \(\eps\)-differentially private. Then the mechanism \(\mech\) that outputs
\[
\mech(\ds) = \queries(\ds) + \frac{1}{\dsize}Z = \widetilde{\queries}(\ds) + \queries(\elem_0) + \frac{1}{\dsize}{Z}  = \widetilde{\mech}(\ds) + \queries(\elem_0)
\]
is a post-processing of \(\widetilde{\mech}\), and is, therefore, also
\(\eps\)-differentially private. Efficiency follows directly from Lemma~\ref{lm:Knorm-noise}.

It remains to analyze the error of the mechanism, i.e., to bound \(\frac{1}{\sqrt{\qsize} \dsize}\E\|Z\|_2\).  By  Lemmas~\ref{lm:Knorm-noise}~and~\ref{lm:vollb-meanw}, we have
\begin{align*}
    \err(\mech, \queries, \dsize) = \frac{1}{\sqrt{\qsize} \dsize}\E\|Z\|_2
    &\lesssim 
\frac{1}{\sqrt{\qsize}n\eps}\left(\sum_{i = 1}^{\lfloor\log_2 \qsize \rfloor}{\left( \sqrt{2^{-i}\qsize} w(\spoly_{\widetilde{\queries}}^\sym)\right)^{2/3}}\right)^{3/2}\\
&= 
\frac{w(\spoly_{\widetilde{\queries}}^\sym)}{n\eps}\left(\sum_{i = 1}^{\lfloor\log_2 \qsize \rfloor}{2^{-i/3}}\right)^{3/2}\\
&\lesssim
\frac{w(\spoly_{\widetilde{\queries}}^\sym)}{n\eps}.
\end{align*}
In the final line, we used that the sum \(\sum_{i = 1}^{\infty}{2^{-i/3}}\) converges. We then observe that, for a standard Gaussian \(G\) in \(\R^\qsize\),
\begin{multline*}
    w(\spoly_{\widetilde{\queries}}^\sym) = 
    \E \sup_{y \in \spoly_{\widetilde{\queries}} \cup -\spoly_{\widetilde{\queries}}} \ip{G}{y}
    = \E \sup_{y \in \spoly_{\widetilde{\queries}}} |\ip{G}{y}|
    \le 2w(\spoly_{\widetilde{\queries}})
    = 2w(\spoly_\queries - \queries(\elem_0)) =
2 w(\spoly_\queries).
\end{multline*}
For the first inequality, we use \(0 \in
\spoly_{\widetilde{\queries}}\), and Exercise~7.6.9 in
Vershynin. (This is where we crucially rely on
\(\spoly_{\widetilde{\queries}}\) containing the origin, which is why
we needed to introduce the modified workload
\(\widetilde{\queries}\).) For the final two equalities we use
\(\spoly_{\widetilde{\queries}} = \spoly_\queries -
\queries(\elem_0)\) and Proposition 7.5.2, part (ii) in
Vershynin. This completes the proof of the main error bound.

The bound for the case
when \(\queries\) is a workload of bounded queries follows from the estimate
\begin{equation}\label{eq:width-bounded}
w(\spoly_\queries)
= \E \sup_{y \in \spoly_\queries} \ip{G}{y}
= \E \sup_{\elem \in \uni} \ip{G}{\queries(\elem)}
\lesssim \sqrt{\qsize \log \usize},
\end{equation}
where \(G\) is a standard Gaussian in \(\R^\qsize\), the second
equality holds because \(\sup_{y \in \spoly_\queries} \ip{G}{y}\) is
achieved at a vertex of \(\spoly_\queries\), and the inequality holds
from standard Gaussian process estimates (see, e.g., Example 7.5.10
in~\cite{Vershynin-HDP}) and because for bounded queries
\(\|\queries(\elem)\|_2 \le \sqrt{\qsize}\) for all \(\elem \in \uni\).
\end{proof}

\section{Instantiation of the JL Mechanism}

\begin{theorem}\label{thm:main}
For any workload of \(\qsize\) queries \(\queries\) over a universe of size \(\usize\), and any \(0 \le \alpha \le \frac{\diam(\spoly_\queries)}{\sqrt{\qsize}}\),
there exists an efficient \(\eps\)-differentially private mechanism \(\mech\) with sample complexity
\[
\sc(\mech,\queries,\alpha) 
\lesssim 
\frac{w(\spoly_\queries)^2}{{\qsize}\alpha^2\eps}.
\]
In particular, if \(\queries\) is a workload of \(\qsize\) bounded queries over a universe of size \(\usize\), we have
\(
\sc(\mech,\queries,\alpha) 
\lesssim \frac{\log \usize}{\eps \alpha^2}.
\)
\end{theorem}
\begin{proof}
The mechanism \(\mech\) applies Lemma~\ref{lm:template} to
\(\queries\), with the noise \(Z\) sampled as in
Lemma~\ref{lm:noise-meanw} applied to the queries \(\jl \queries\),
where \(\jl\) is any JL-matrix with parameter \(C \lesssim 1\). By
Lemmas~\ref{lm:template}~and~\ref{lm:noise-meanw}, \(\mech\) is
\(\eps\)-differentially private, efficient, and its error is bounded as 
\[
\err(\mech,\queries,n)
\lesssim
\frac{w(\spoly_\queries)}{\sqrt{\qsize\jldim}} + \frac{\sqrt{\ell}}{\sqrt{k}n\eps} \E w(\jl \spoly_\queries)
\lesssim \frac{w(\spoly_\queries)}{\sqrt{\qsize\jldim}} + \frac{\sqrt{\ell}w(\spoly_\queries)}{\sqrt{k}\dsize\eps}, %
\]
where we used Lemma~\ref{lm:jl-gauss} in the second inequality. The
right hand side is minimized at \(\ell = {\dsize \eps}\), giving
error
\[
\err(\mech,\queries,n)
\lesssim
\frac{w(\spoly_\queries)}{\sqrt{\qsize n \eps}},
\]
and the sample complexity bound follows from setting \(\dsize\) so
that the right hand side is at most \(\alpha\). 
The bound for the case
when \(\queries\) is a workload of bounded queries follows from the estimate
\eqref{eq:width-bounded}.
\end{proof}

Theorem~\ref{thm:meanw} follows directly from
Lemma~\ref{lm:noise-meanw}, and Theorem~\ref{thm:main}. 
Theorem~\ref{thm:main-worst-case} follows from
Lemmas~\ref{lm:ball-noise}~and~\ref{lm:noise-meanw}, and
Theorem~\ref{thm:main}, since a separation oracle for
\(\spoly_\queries\) can be implemented in time polynomial in
\(\usize\) and \(\qsize\) using linear programming. 

\section{Efficient Mechanisms for 2-way Marginals and Covariance Estimation}

As an application, we consider the problem of estimating the
covariance of a distribution \(\mu\) on \(B_2^d\) in Frobenius norm,
and prove Theorem~\ref{thm:cov}. Note that in this section we do not
necessarily normalize our error.

Theorem~\ref{thm:cov}, as well as an efficient mechanism for computing
2-way marginals, follow from the following theorem.

\begin{lemma}\label{lm:emp-cov}
  There exists an \(\eps\)-differentially private mechanism \(\mech\)
  such that,
  given a dataset \(\ds = (x_1x_1^T, \ldots, x_\dsize x_\dsize^T)\)
  with \(x_1, \ldots, x_n \in rB_2^d\), \(\mech\) outputs a \(d\times d\) matrix
  \(\mech(\ds)\) satisfying
  \[
    \E \left\| \frac{1}{\dsize}\sum_{i = 1}^\dsize x_i x_i^T -
      \mech(\ds)\right\|_F \le \alpha r
  \]
  as long as
  \[
    \dsize \ge C\min\left\{\frac{d^{1.5}}{\alpha\eps},
          \frac{d}{\alpha^2\eps}\right\}
  \]
  for a constant \(C\).
  Moreover, the mechanism runs in time polynomial in
  \(d,\dsize,\frac{1}{\eps}\). 
\end{lemma}
\begin{proof}
  Let us denote by
  \( M_2(X) = \frac{1}{\dsize}\sum_{i = 1}^\dsize x_i x_i^T \) the
  second moment matrix of \(\ds\).  We notice that \(M_2(X)\) is
  equivalent to a workload of \(d^2\) statistical queries on the
  universe \(B_2^d\). The sensitivity polytope (or, more appropriately
  in this case, sensitivity body) of \(M_2\) is
  \(K_{M_2} = \mathrm{conv}\{xx^T: x \in rB_2^d\}\). To estimate \(M_2\), we use the better 
  mechanism among the one in Lemma~\ref{lm:noise-meanw}, and
  the one in Theorem~\ref{thm:main}. Since the Frobenius norm is simply the
  \(\ell_2\) norm on matrices, seen as elements in $\R^{d^2}$,
  Lemma~\ref{lm:noise-meanw} and
  the one in Theorem~\ref{thm:main} imply that we can achieve
  \(\E\|M_2(\ds) - \mech(\ds)\|_F \le \alpha r\) as long as
  \begin{equation}\label{eq:cov-sc}
    \dsize \ge C\min\left\{\frac{dw(K_{M_2})}{\alpha r
        \eps},\frac{w(K_{M_2})^2}{\alpha^2 r^2 \eps}\right\},
  \end{equation}
  for a sufficiently large constant \(C\). To show that this implies
  the sample complexity bound in the statement of the theorem, we need
  to bound the Gaussian width \(w(K_{M_2})\). Noticing that the standard
  inner product \(\ip{A}{B}\) of two matrices \(A\) and \(B\) can be
  written as \(\tr(AB^T)\), we get
  \[
    w(K_{M_2}) = \E \sup_{x \in rB_2^d} \tr(Gxx^T) = r^2 \E \sup_{x \in B_2^d}
    x^T G x = r^2 \E\|G\|,
  \]
  with \(G\) a \(d\times d\) matrix of independent standard Gaussians,
  and \(\|G\|\) denoting its operator norm. It is a classical fact
  (see, e.g., Theorem 4.4.5 in \cite{Vershynin-HDP}) that
  \(\E \|G\| \lesssim \sqrt{d}\), and this, together with
  \eqref{eq:cov-sc} finishes the proof of the sample complexity bound.

  Finally, to prove efficiency, we need to exhibit a separation oracle
  for \(K_{M_2} = \mathrm{conv}\{xx^T: x \in rB_2^d\}\). We claim the alternative
  characterization \(K_{M_2} = \{A \in \R^{d\times d}: A \succeq 0, \tr A
  = r^2\}\), where the notation \(A \succeq 0\) means that \(A\) is
  positive semidefinite. On the one hand, it is clear that for any \(A
  \in K_{M_2}\), \(A\succeq 0\) and \(\tr A = r^2\), since these conditions
  hold for any extreme point \(xx^T\), \(x\in rB_2^d\), of \(K_{M_2}\), and
  are maintained under taking convex combinations. On the other hand, if
  \(A \succeq 0\) and \(\tr A = r^2\), then, by the spectral theorem,
  \(A =  \sum_{i = 1}^d \lambda_i (ra_i) (ra_i)^T\), where the \(\lambda_i\ge
  0\) are the eigenvalues of \(\frac{1}{r^2} A\), and \(a_i \in B_2^d\) are the
  eigenvectors. Moreover \(\sum_i \lambda_i = \frac{\tr(A)}{r^2} = 1\), so this
  spectral decomposition in fact gives a representation of \(A\) as
  convex combination of points in \(\{xx^T: x\in rB_2^d\}\). With this
  alternative characterization, we see that \(K_{M_2}\) can be written
  as a feasibility SDP, so it has a polynomial time separation oracle.
\end{proof}

First we use Lemma~\ref{lm:emp-cov} to prove Theorem~\ref{thm:cov}.

\begin{proof}[Proof of Theorem~\ref{thm:cov}]
A standard symmetrization argument shows that, if the  samples from \(\mathcal{D}\) are
  \(X = (X_1, \ldots, X_\dsize)\),  and \({\Sigma(X)}\) is
  their empirical covariance matrix, then
  \[
    \E\| {\Sigma}(X) - \Sigma\|_F \lesssim \frac{1}{\sqrt{\dsize}}.
  \]
  This means that the non-private sample complexity of the problem is
  at most on the order of \(\frac{1}{\alpha^2}\), and it suffices to
  privately estimate the empirical covariance within error \(\alpha\)
  with sample complexity
  \(
  \min\left\{\frac{d^{1.5}}{\eps \alpha}, \frac{d}{\eps
      \alpha^2}\right\}
  \).
  To do so, we estimate the empirical second moment matrix
  \(
  {M}_2(X) = \frac{1}{\dsize}\sum_{i = 1}^\dsize X_i X_i^T,
  \)
  within Frobenius error \(\frac{\alpha}{2}\),
  and the empirical mean
  \(
  {\mu}(X) = \frac{1}{\dsize}\sum X_i
  \)
  within \(\ell_2\) error \(\frac{\alpha}{4}\), each   under
  \(\frac{\eps}{2}\)-differential privacy. Denoting the
  estimates \(\widetilde{M}_2\) and \(\widetilde{\mu}\), respectively,
  and denoting \(\widetilde{\Sigma} = \widetilde{M}_2 - \widetilde{\mu}\widetilde{\mu}^T\),  this gives us
  \begin{align*}
    \|\widetilde{\Sigma} - \Sigma\|_F
    &\le
      \|\widetilde{M}_2 - M_2\|_F + \|\widetilde{\mu}\widetilde{\mu}^T -
      \mu\mu^T\|_F\\
    &\le    \|\widetilde{M}_2 - M_2\|_F +
      \|\widetilde{\mu}(\widetilde{\mu}^T -\mu^T)\|_F +
      \|(\widetilde{\mu}-\mu)\mu^T\|_F\\
    &= \|\widetilde{M}_2 - M_2\|_F +
      (\|\widetilde{\mu}\|_2 + \|\mu\|_2)\|\widetilde{\mu}^T -\mu^T\|_2\\
    &\le  \|\widetilde{M}_2 - M_2\|_F + 2\|\widetilde{\mu} - \mu\|_2,
  \end{align*}
  where in the last inequality we used the fact that \(\mathcal{D}\)
  is supported on \(B_2^d\), so \(\|\mu\|_2 \le 1\). Therefore, as
  long as \(\E\|\widetilde{M}_2 - M_2\|_F \le \frac{\alpha}{2}\), and
  \(\E\|\widetilde{\mu} - \mu\|_2\le \frac{\alpha}{4}\), our
  mechanism can achieve the desired error and privacy guarantees if
  it outputs \(\widetilde{\Sigma}\).

  It follows from Lemma~\ref{lm:emp-cov} (with \(r=1\)) that the
  estimate \(\widetilde{M}_2\) of
  \(M_2(\ds)\) can be computed under \(\frac{\eps}{2}\)-differential
  privacy with error
  \(\E\|\widetilde{M}_2 - M_2\|_F \le \frac{\alpha}{2}\) and 
  with the required sample complexity
  \( \min\left\{\frac{d^{1.5}}{\eps \alpha}, \frac{d}{\eps
      \alpha^2}\right\}.  \) We then notice that
  \(\mu\) is a workload of \(d\) statistical queries on the universe,
  \(B_2^d\), and its sensitivity polytope (i.e., body) is
  \(K_\mu = B_2^d\). We use the mechanism in Lemma~\ref{lm:ball-noise}
  to compute \(\widetilde{\mu}\). The lemma guarantees that we achieve
  \(\frac{\eps}{2}\)-differential privacy and error
  \(\E\|\widetilde{\mu} - \mu\|_2\le \frac{\alpha}{4}\) as long as
  \(\dsize \ge \frac{Cd}{\eps\alpha}\) for a sufficiently large
  constant \(C\). The mechanism is also efficient, as it only needs a
  separation oracle for \(B_2^d\).
\end{proof}

Finally, we observe that Lemma~\ref{lm:emp-cov} implies an efficient
algorithm for estimating \(2\)-way marginals with optimal sample
complexity. Recall that \(Q_{2,d}\) is the workload of 2-way marginals
on the universe \(\{0,1\}^d\), consisting of the \({d \choose 2}\) queries \(q_{\{i,j\}}(x)
= x_i x_j\), \(i,j \in [d], i\neq j\). 

\begin{theorem}\label{thm:efficient-2way}
  For any integer \(d\ge 1\), there exists an \(\eps\)-differentially
  private algorithm \(\mech\) whose sample complexity for the workload
  of \(2\)-way marginals is
  \[
    \sc(\mech,Q_{2,d},\alpha) \lesssim
    \min\left\{\frac{d^{1.5}}{\alpha\eps},
      \frac{d}{\alpha^2\eps}\right\}.
  \]
  Moreover, the mechanism runs in time polynomial in
  \(d,\dsize,\frac{1}{\eps}\). 
\end{theorem}
\begin{proof}
  Note that for a dataset \(\ds = (x_1,\ldots, x_n)\), and any distinct
  \(i,j \in [d]\),
  \[
    q_{\{i,j\}}(\ds) = \left(\frac{1}{\dsize}\sum_{i = 1}^\dsize x_i x_i^T\right)_{i,j}.
  \]
  Therefore, to approximate \(Q_{2,d}\) it suffices to run the
  mechanism \(\mech\) from Lemma~\ref{lm:emp-cov}, and the
  (normalized) error we get is
  \[
    \E\left(\frac{1}{{d\choose 2}} \sum_{1 \le i < j \le d}
      \left(q_{i,j}(\ds) -\mech(\ds)_{i,j}\right)^2
    \right)^{1/2}
    \lesssim
    \E \frac{1}{d} \|M_2(\ds) - \mech(\ds)\|_F,
  \]
  where \(M_2(\ds) = \frac{1}{\dsize}\sum_{i = 1}^\dsize x_i
  x_i^T\). Since \(\{0,1\}^d \subseteq \sqrt{d}B_2^d\),
  Lemma~\ref{lm:emp-cov} implies the required sample complexity of
  \(
  \min\left\{\frac{d^{1.5}}{\alpha\eps},
    \frac{d}{\alpha^2\eps}\right\},
  \)
  as well as the running time guarantee.
\end{proof}

Note that the sample complexity in Theorem~\ref{thm:efficient-2way}
matches the bound in \eqref{eq:sc-marg-opt} with a mechanism running
in time polynomial in \(d\). Moreover, as we will see in
Section~\ref{sect:marg-lb}, this sample complexity bound is optimal
among all \(\eps\)-differentially private mechanisms. There is no
differentially private mechanism known to achieve the optimal sample
complexity for \(w\)-way marginals in time polynomial in \(d^w\) for
any \(w \ge 3\), and also no complexity theoretic evidence that this
is impossible. 

\section{Lower Bounds}

In this section we show that the sample complexity bound in 
Theorem~\ref{thm:main-worst-case} is tight both for random query
workloads, and for \(w\)-way marginals.

\subsection{Worst Case Queries}\label{sect:rand-lb}

The following lower bound shows that Theorem~\ref{thm:main-worst-case}
is tight up to constants in the worst case.

\begin{theorem}\label{thm:lb-worst-case}
  There exists a constant \(c > 0\) such that for any \(\alpha \in
  (0,c)\), \(\eps \in (0,1)\),  and any \(\qsize\) and \(\usize\) such
  that \(\qsize \le
  \frac12 \usize\) and \(\usize \ge \frac{2}{\alpha^2}\),  there exists a  workload \(\queries\) of
  \(\qsize\) bounded queries on a universe of size \(\usize\) for which
  \begin{equation}
    \label{eq:sc-opt-lb}
    \sc_\eps(\queries,\alpha)
    \gtrsim
    \min\left\{
    \frac{\qsize}{\eps \alpha},
    \frac{\sqrt{\qsize \log(\usize/\qsize)}}{\eps \alpha},
    \frac{\log(\alpha^2 \usize)}{\eps \alpha^2}
    \right\}.
\end{equation}
\end{theorem}
\begin{proof}
  The first two terms on the right hand side of \eqref{eq:sc-opt-lb}
  follow from results of Hardt and Talwar~\cite{HardtT10}, and
  De~\cite{De12}: see, in particular, Theorem 3.2 in~\cite{De12}. To
  establish the third term, we slightly modify De's construction. By
  Claim 3.1 in~\cite{De12}, and Lemma~\ref{lm:packing-general}, it is enough to construct datasets
  \(\ds_1, \ldots, \ds_M \in \uni^\dsize\) such that
  \begin{itemize}

  \item \(\dsize = \frac{\log(M/2)}{\eps} \gtrsim  \frac{\log(\alpha^2 \usize)}{\eps\alpha^2};
  \)
  
  \item any two distinct \(\ds_i, \ds_j\) differ in at least \(\Delta
    \gtrsim n\) elements;

  \item there is an \(a \gtrsim \alpha^2 n\) such that for any \(\ds_i\) and any \(\elem \in
    \uni\), \(\elem\) appears either \(a\) or \(0\)
    times in \(\ds_i\).

  \end{itemize}
  Similarly to Claim~A.1 in~\cite{De12}, the
  existence of such \(\ds_1, \ldots, \ds_M\) follows from classical
  lower bounds on sizes of packings. We let \(\dsize =
  \frac{\log(M/2)}{\eps}\), and define  \(s = \frac{\dsize}{a}\) and \(t = \frac{\dsize -
    \Delta}{a}\). Then, for appropriate choices of \(a \gtrsim
  \alpha^2 \dsize\) and \(\Delta \gtrsim \dsize\), 
  Proposition~2.1 in~\cite{EFF-packings} shows that there exists a
  family \(S_1, \ldots, S_M\) of subsets of \(\uni\) such that \(|S_i|
  = s\) for all \(i\), \(|S_i \cap S_j| \le t\) for all \(i \neq
  j\), and
  \[
    \log(M/2) \gtrsim  t \log\left(\frac{\usize t}{s^2}\right)
    \gtrsim \frac{\log(\alpha^2 \usize)}{\alpha^2}.
  \]
  To finish the construction of the datasets, we define \(\ds_i\)
  to contain \(a\) copies of each element of \(S_i\). It is then
  straightforward to verify that   \(\ds_1, \ldots, \ds_M \in
  \uni^\dsize\) satisfy the required properties.
\end{proof}

It is worth noting that the lower bound in
Theorem~\ref{thm:lb-worst-case} holds with high probability for random
\(\pm 1\)-valued queries, i.e., queries where for each
\(\elem \in \uni\), \(\queries(\elem)\) is an independently drawn
uniformly random vector in \(\{-1,+1\}^\qsize\).

\subsection{Marginals}\label{sect:marg-lb}

Except for random queries, 
Theorem~\ref{thm:main-worst-case} also gives tight bounds for the
important class of constant width marginal queries. The next theorem
uses a packing argument to show that the sample complexity bound \eqref{eq:sc-marg-opt} is tight.
\begin{theorem}\label{thm:marg-lb}
For any integer \(w \ge 1\) there exists a constant \(c>0\) such that
for any integer \(d \ge 10w\) and for any \(\alpha \in (0,c]\), we have
\[   \sc_\eps(\queries_{w,d},\alpha)
    \ge c
    \min\left\{
    \frac{d^w}{\eps \alpha},
    \frac{d^{(w+1)/2}}{\eps \alpha},
    \frac{d}{\eps \alpha^2}
    \right\}.
\]
\end{theorem}
\begin{proof}
  The first term in the lower bound only achieves the minimum when $w
  = 1$, so we only have to prove it holds in the case of 1-way
  marginals. This was shown by Hardt and Talwar~\cite{HardtT10}; let
  us briefly sketch the proof. Note that the set \(S_{1,d} =
  \{\queries_{1,d}(\elem): \elem \in \{0,1\}^d\}\) is simply the
  boolean cube \(\{0,1\}^d\). By classical packing results
  (e.g.~\cite{EFF-packings}),
  \(
  \log \sep(\{0,1\}^d, \frac{\sqrt{d}}{2} B_2^d) \gtrsim d.
  \)
  Then the lower bound follows from Lemma~\ref{lm:packing-singleton}.

  For the other two terms, we build packings of datasets and use
  Lemma~\ref{lm:packing-general}. We start with the second term, which
  achieves the minimum when \(w \ge 2\), and \(\alpha \le
  d^{-(w-1)/2}\).\cut{Note that
  \(S_{w,d} = \{\queries_{w,d}(\elem): \elem \in \{0,1\}^d\}\) is a
  coordinate projection of \(\{\elem^{\otimes w}: \elem\in \{0,1\}^d\}\).} Let
  \(h_0, h_1, \ldots, h_{d-1}\in \{-1,1\}^d\) be columns of a Hadamard
  matrix (i.e., any two distinct \(h_i,h_j\) are orthogonal), where
  \(h_0\) is the all-ones vector. Let us define, for \(i \in [d-1]\),
  \(g_i = \frac12 (h_0 + h_i) \in \{0,1\}^d\). Let, 
  \(y_1, \ldots, y_M \in \{0,1\}^d\) be such that for all distinct
  \(y_i, y_j\) we have \(\|y_i - y_j\|_2 \ge \frac{\sqrt{d}}{2}\). As
  noted above, by standard packing results, we can take
  \(\log M \gtrsim d\). Let us now define, for \(i \in [M]\), and \(s
  \in [d-1]^{w-1}\), \(\elem_{i, s} \in \{0,1\}^{wd}\)
  to be the concatenation of \(y_i\) and \(g_{s_1}, \ldots,
  g_{s_{w-1}}\). For any \(f:[d-1]^{w-1} \to [M]\), we define a
  dataset \(\ds_f\) that consists of \(a \gtrsim \frac{d}{\eps}\) copies of each of
  the elements \(\elem_{f(s),s}\)
  for each sequence \(s \in [d-1]^{w-1}\), together with \(\dsize -
  a(d-1)^{w-1}\) copies of the zero vector, where \(\dsize \ge a(d-1)^{w-1}\) will
  be chosen later. We consider a family \(\cal F\) of functions from
  \([d-1]^{w-1}\) to \([M]\) such that for any two distinct \(f,f' \in
  \mathcal{F}\) we have \(|\{s \in [d-1]^{w-1}:f(s) \neq
  f'(s)\}|\gtrsim (d-1)^{w-1}\). Again by standard packing results, we
  can find such a family with \(\log |\mathcal{F}| \gtrsim (d-1)^{w-1}
  \log M\).  The packing of datasets we consider consists of
  \(\ds_f\) for \(f \in \mathcal{F}\).

  An important observation we make is that if we project
  \(\queries_{w,wd}(\elem_{i,s})\) onto coordinates corresponding to queries
  \(\query_S\) for which \(S\) has 1 element from the first \(d\)
  coordinates, one from the second \(d\) coordinates, etc., then we get the vector
  \(y_i \otimes \bigotimes_{j =1}^{w-1} g_{s_j}\). We then have
  \begin{align}
    \|\queries_{w,wd}(\ds_f) -  \queries_{w,wd}(\ds_{f'})\|_2^2
    &\ge
    \frac{a^2}{\dsize^2} \left\|\sum_{s \in [d-1]^{w-1} }  (y_{f(s)}-y_{f'(s)}) \otimes
      \bigotimes_{j =1}^{w-1} g_{s_j} \right\|_2^2\notag\\
    &=
      \frac{a^2}{2^{2w-2}\dsize^2}
      \left\|\sum_{s \in [d-1]^{w-1}}  (y_{f(s)}-y_{f'(s)}) \otimes
      \bigotimes_{j =1}^{w-1}  (h_0 + h_{s_j}) \right\|_2^2\label{eq:marg-dist-first}.
  \end{align}
  Let \(\Pi\) be the orthogonal projection matrix onto the hyperplane
  orthogonal to \(h_0\) (i.e., to the all-ones vector). Then
  multiplying
  \((y_{f(s)}-y_{f'(s)}) \otimes
  \bigotimes_{j =1}^{w-1}  (h_0 + h_{s_j})\)
  by the orthogonal projection matrix \(I \otimes \Pi^{\otimes w-1}\)
  gives
  \[
    (y_{f(s)}-y_{f'(s)}) \otimes
    \bigotimes_{j =1}^{w-1} \Pi(h_0 + h_{s_j}) =
    (y_{f(s)}-y_{f'(s)}) \otimes
    \bigotimes_{j =1}^{w-1} h_{s_j}.
  \]
  Since orthogonal projection cannot increase
  the \(\ell_2\) norm, this implies
  \begin{align}
    \left\|\sum_{s \in [d-1]^{w-1}}  (y_{f(s)}-y_{f'(s)}) \otimes
    \bigotimes_{j =1}^{w-1}  (h_0 + h_{s_j}) \right\|_2^2
    &\ge
    \left\|\sum_{s \in [d-1]^{w-1}}  (y_{f(s)}-y_{f'(s)}) \otimes
      \bigotimes_{j =1}^{w-1}  h_{s_j} \right\|_2^2\notag\\
    &=
      \sum_{s \in [d-1]^{w-1}}\left\|(y_{f(s)}-y_{f'(s)}) \otimes
      \bigotimes_{j =1}^{w-1}  h_{s_j} \right\|_2^2\label{eq:marg-orthog}\\
    &=
      \sum_{s \in [d-1]^{w-1}}d^{w-1}\|y_{f(s)}-y_{f'(s)}\|_2^2.\label{eq:marg-dist}
  \end{align}
  Equality \eqref{eq:marg-orthog} holds since for any two distinct \(s,s' \in [d-1]^{w-1}\)
  \[
    \left \langle (y_{f(s)}-y_{f'(s)}) \otimes \bigotimes_{j =1}^{w-1} h_{s_j},
      (y_{f(s')}-y_{f'(s')}) \otimes \bigotimes_{j =1}^{w-1} h_{s'_j}\right\rangle
    = \ip{y_{f(s)}-y_{f'(s)}}{y_{f(s')}-y_{f'(s')}}\prod_{j = 1}^{w-1}
    \ip{ h_{s_j}}{ h_{s'_j}} = 0,
  \]
  with the final equality due to the pairwise orthogonality of \(h_1,
  \ldots, h_{d-1}\). Going back to \eqref{eq:marg-dist}, we recall
  that for any \(f,f' \in \mathcal{F}\), the number of sequences \(s\)
  for which \(f(s) \neq f'(s)\) is \(\gtrsim (d-1)^{w-1}\). Since any
  two distinct \(y_i, y_j\) satisfy \(\|y_i - y_j\|_2^2 \gtrsim d\),
  this gives us
  \begin{equation}
    \label{eq:marg-dist-final}
    \sum_{s \in [d-1]^{w-1}}d^{w-1}\|y_{f(s)}-y_{f'(s)}\|_2^2
    \gtrsim d^{2w-1}. 
  \end{equation}
  Combining \eqref{eq:marg-dist-first}, \eqref{eq:marg-dist}, and
  \eqref{eq:marg-dist-final}, we get for any two distinct \(f,f' \in \mathcal{F}\),
  \[
    \frac{1}{\sqrt{{wd \choose w}}}\|\queries_{w,wd}(\ds_f) -  \queries_{w,wd}(\ds_{f'})\|_2
    \gtrsim
    \frac{c ad^{(w-1)/2}}{\dsize}.
  \]
  where \(c\) is a constant that may depend on \(w\). Towards applying
  Lemma~\ref{lm:packing-general}, let us choose \(a\gtrsim
  \frac{d}{\eps}\) so that \(a(d-1)^{w-1} \le
  \frac{\log(|\mathcal{F}|/2)}{\eps}\). Then, 
  \[
    \frac{1}{\sqrt{{wd \choose w}}}\|\queries_{w,wd}(\ds_f) -  \queries_{w,wd}(\ds_{f'})\|_2
    \gtrsim
    \frac{c d^{(w+1)/2}}{\eps\dsize},
  \]
  and, since any two distinct
  datasets \(\ds_f, \ds_{f'}\) differ in at most \(a(d-1)^{w-1}\)
  elements, we can apply   Lemma~\ref{lm:packing-general} with \(\dsize = \frac{c
    d^{(w+1)/2}}{\alpha\eps}\) to get the lower bound. Recall that we
  need \(\dsize \ge a(d-1)^{w-1}\), which happens when \(\alpha \le c'
  d^{-(w-1)/2}\) for a sufficiently small \(c' > 0\). This is exactly
  the parameter regime in which we need the second term of the lower
  bound to hold. 
  Note that we proved the
  lower bound against \(\queries_{w,wd}\), but this implies a lower
  bound for \(\queries_{w,d}\) for \(d \ge w\).

  The third term of the lower bound is proved analogously to the
  second, with some modifications. We choose some subset \(\Sigma
  \subseteq [d-1]^{w-1}\) of size \(|\Sigma|\gtrsim
  \frac{1}{\alpha^2}\) and only consider functions \(f:\Sigma \to
  [M]\). When defining a dataset \(\ds_f\), we take \(a \gtrsim
  \frac{d}{\eps}\) copies of each of
  the elements \(\elem_{f(s),s}\)
  for each sequence \(s \in \Sigma\), and we do not take any
  additional copies of the zero vector. The rest of the proof proceeds
  mutatis mutandis as the proof of the second term of the lower
  bound. 
\end{proof}


\section{Optimality of the Mechanisms}

 In the previous section we argued that the mechanisms in this paper
 are optimal for worst-case query workloads. We can in fact show that
 in the constant error regime the sample complexity of our mechanisms
 is also tight with respect to the optimal sample complexity for the
 given workload. The next lemma follows from combining
 Theorem~\ref{thm:meanw} with the methods of Blasiok, Bun, Nikolov,
 and Steinke~\cite{cdp}, and is the main ingredient in proving
 Theorem~\ref{thm:opt}

 \begin{lemma}\label{lm:opt}
   For any workload \(\queries\) of \(\qsize\) queries over a universe of
   size \(\usize\), and any \(0\le \alpha \le
   \frac{\diam(\spoly_\queries)}{\sqrt{\qsize}}\), there exists an
   \(\eps\)-differentially private mechanism \(\mech\) with running time
   polynomial in \(\qsize, \usize, \dsize, \frac{1}{\eps}\), and with sample complexity
   \begin{equation}\label{eq:opt-ub}
     \sc(\mech,\queries,\alpha)
     \lesssim
     \frac{1}{\eps\alpha^2} \log\left(\frac{\diam(S_\queries)}{\alpha \sqrt{\qsize}}\right)^2 \sup_{t \ge \alpha/4} t^2 \log \sep(S_\queries, t\sqrt{\qsize}B_2^\qsize),
   \end{equation}
   where \(S_\queries = \{\queries(\elem):\elem \in \uni\}\).
 \end{lemma}

 In the proof of Lemma~\ref{lm:opt}, we use the following lemma
 from~\cite{cdp}.
 \begin{lemma}\label{lm:dudley}
   Let \(S \subseteq \R^\qsize\) be a set such that for any distinct
   \(y,y' \in S\) we have \(\|y-y'\|_2 >
   \frac{\alpha\sqrt{\qsize}}{2}\). Then
   \[
     w(S) \lesssim
     \sqrt{\qsize}\log\left(\frac{\diam(S)}{\alpha\sqrt{\qsize}}\right)
     \sup\left\{t\sqrt{\log \sep(S,t\sqrt{\qsize}B_2^\qsize)}:t \ge \frac{\alpha}{4}\right\}.
   \]
 \end{lemma}

The mechanism in Lemma~\ref{lm:opt} is given by
Algorithm~\ref{alg:coarse}.

\begin{algorithm}
  \begin{algorithmic}
    \State Let \(S \subseteq S_\queries = \{\queries(\elem):\elem \in
    \uni\}\) be inclusion-maximal s.t.~for all distinct \(y,y' \in S\), \(\|y-y'\|_2 >
    \frac{\alpha \sqrt{\qsize}}{2}\).

    \State Define workload \(\widetilde{\queries}\) s.t.~for each \(\elem \in
    \uni\),
    \(
    \widetilde{\queries}(\elem) \in \arg \min\{\|y -
    \queries(\elem)\|_2: y \in S\}.
    \)

    \State Apply the mechanism from Theorem~\ref{thm:meanw} to \(\widetilde{\queries}(\ds)\),
    and output answer.
  \end{algorithmic}
  \caption{The Coarse JL-Release mechanism}
  \label{alg:coarse}
\end{algorithm}

\begin{proof}[Proof of Lemma~\ref{lm:opt}]
  As already mentioned, the required mechanism \(\mech\) is given by
  Algorithm~\ref{alg:coarse}. The first key claim is that, for all
  \(\ds \in \uni^\dsize\),
  \begin{equation}\label{eq:coarse-round}
    \frac{1}{\sqrt{\qsize}}\|\widetilde{\queries}(\ds) - \queries(\ds)\|_2 \le \frac{\alpha}{2}.
  \end{equation}
  By the triangle inequality, in order to verify
  \eqref{eq:coarse-round}, it suffices to check it for single-element
  datasets, i.e., to check that, for every \(\elem \in \uni\),
  \(
  \frac{1}{\sqrt{\qsize}}\|\widetilde{\queries}(\elem) - \queries(\elem)\|_2 \le \frac{\alpha}{2}.
  \)
  This follows by maximality of \(S\): if there were any \(\elem \in
  \uni\) such that
  \(
  \frac{1}{\sqrt{\qsize}}\|\widetilde{\queries}(\elem) - \queries(\elem)\|_2
  > \frac{\alpha}{2},
  \)
  then, by the construction of \(\widetilde{\queries}\), it would
  follow that
  \(
  \frac{1}{\sqrt{\qsize}}\|y - \queries(\elem)\|_2
  > \frac{\alpha}{2}
  \)
  for all \(y \in S\), and we could then add \(\queries(\elem)\) to
  \(S\), contradicting its maximality. Therefore, no such \(\elem\)
  exists, and this proves \eqref{eq:coarse-round}.

  By Theorem~\ref{thm:meanw}, we have
  \(
  \E \frac{1}{\sqrt{\qsize}} \|\mech(\ds) -
  \widetilde{\queries}(\ds)\|_2 \le \frac{\alpha}{2}
  \)
  as long as
  \(
  n \ge \frac{C w(\mathrm{conv}(S))^2}{\qsize\eps\alpha^2}
  = \frac{C w(S)^2}{\qsize\eps\alpha^2}
  \)
  for a constant \(C>0\).
  Together with Lemma~\ref{lm:dudley} and \eqref{eq:coarse-round},
  this means that
  \[
  \sc(\mech,\queries,\alpha)
  \lesssim
  \frac{w(S)^2}{\qsize \alpha\eps^2}
  \lesssim
  \frac{1}{\eps\alpha^2}\log\left(\frac{\diam(S)}{\alpha\sqrt{\qsize}}\right)^2
  \sup\left\{t^2{\log \sep(S,t\sqrt{\qsize}B_2^\qsize)}:t \ge \frac{\alpha}{4}\right\}.
  \]
  This establishes the required sample complexity, since
  \(\sep(S,t\sqrt{\qsize}B_2^\qsize) \le \sep(S_\queries, t\sqrt{\qsize}B_2^\qsize)\)

  The running time guarantee holds since \(S\) can be computed
  greedily by making a single pass over \(\uni\), and the rest of the
  algorithm is just the mechanism from Theorem~\ref{thm:meanw}. Privacy follows from Theorem~\ref{thm:meanw}.
\end{proof}

\begin{proof}[Proof of Theorem~\ref{thm:opt}]
  By Lemma~\ref{lm:packing-singleton}, we have
  \begin{equation}\label{eq:opt-lb}
    \sc_\eps(\queries,\alpha/8)
    \gtrsim
    \frac{1}{\eps\alpha}
    \sup \left\{t \log \sep(\spoly_\queries,
      t\sqrt{\qsize}B_2^\qsize):t \ge \frac{\alpha}{4}\right\}.
  \end{equation}
  Note that the suprema in \eqref{eq:opt-ub} and \eqref{eq:opt-lb} can
  both be taken over \(t \in
  \left[\frac{\alpha}{4},\frac{\diam(S_\queries)}{\sqrt{\qsize}}\right)\). Therefore, we
  have
  \[
    \frac{\sc(\mech,\queries,\alpha)}{\sc_\eps(\queries,\alpha/8)}
    \lesssim
    \frac{\diam(S_\queries)}{\alpha\sqrt{\qsize}}\log\left(\frac{\diam(S_\queries)}{\alpha
        \sqrt{\qsize}}\right)^2
    =
    \frac{\diam(\spoly_\queries)}{\alpha\sqrt{\qsize}}\log\left(\frac{\diam(\spoly_\queries)}{\alpha
        \sqrt{\qsize}}\right)^2.
  \]
  Together with the running time and privacy guarantees in
  Lemma~\ref{lm:opt}, this proves the theorem.
\end{proof}

\bibliographystyle{alpha}
\bibliography{JLQueryRelease}



\end{document}